\documentclass[11pt]{amsart}
\usepackage[letterpaper,margin=1.08in]{geometry}
\usepackage{amssymb}
\usepackage{dsfont}
\usepackage{microtype}
\usepackage{enumitem}
\usepackage[svgnames]{xcolor}
\usepackage{tikz-cd}
\usepackage[unicode]{hyperref}

\numberwithin{equation}{section}
\newtheorem{theorem}{Theorem}[section]
\newtheorem{lemma}[theorem]{Lemma}
\newtheorem{proposition}[theorem]{Proposition}
\newtheorem{corollary}[theorem]{Corollary}

\title[A linear-in-$q$ range for the MDS conjecture]{A linear-in-$q$ range of dimensions
for the MDS conjecture over $\mathbb F_q$ in odd characteristic}
\author{Xiang Fan}
\address{School of Mathematics, Sun Yat-sen University, Guangzhou 510275, China}
\email{fanx8@mail.sysu.edu.cn}
\date{}
\subjclass[2020]{Primary 51E21, 94B05; Secondary 05B20, 94B27}
\keywords{MDS conjecture, arcs, inclusion matrices, determinant relations,
positive characteristic, normal rational curves, Reed--Solomon codes}
\hypersetup{
  colorlinks=true,
  linkcolor=Navy,
  citecolor=Navy,
  urlcolor=Navy,
  filecolor=Navy,
  breaklinks=true,
  pdftitle={A linear-in-q range of dimensions for the MDS conjecture over Fq in odd characteristic},
  pdfauthor={Xiang Fan},
  pdfsubject={A linear-in-q range of dimensions for the MDS conjecture over Fq in odd characteristic}
}

\begin{document}
\begin{abstract}
Let $q$ be a power of an odd prime $p$.  We prove the MDS conjecture over
$\mathbb F_q$ in every dimension $k$ satisfying
\[
 2\leqslant k\leqslant B(p,q)
 \quad\text{or}\quad
 q+2-B(p,q)\leqslant k\leqslant q,
 \qquad
 B(p,q)=\left\lfloor\frac{(p-2)q+6p-10}{2p-3}\right\rfloor.
\]
For fixed $p$, the first interval is linear in $q$; together with duality,
the theorem covers an asymptotic proportion $1-1/(2p-3)$ of all dimensions.

The proof rests on a vanishing theorem for determinant relations over an
arbitrary field of characteristic $p>0$.  It yields full row rank for
Chowdhury's matrices over a larger range of arc sizes.  In particular, at
$|G|=2k-3+n$ it proves Chowdhury's full-row-rank conjecture without the
$q$-dependent restriction.  Specialization to $\mathbb F_q$, together with
the Ball--Lavrauw construction, gives the stated MDS range.

We also prove that, for every odd prime power $q$, every normal rational
curve in $\mathrm{PG}(N,q)$ is complete for $2\leqslant N\leqslant q-2$,
and every projective Reed--Solomon code of length $q+1$ and dimension
$2\leqslant k\leqslant q-2$ has covering radius $q-k$.
\end{abstract}
\maketitle

\section{Introduction}

The MDS conjecture is a central extremal problem in coding theory and
finite geometry; see \cite[Section~5]{BallLavrauw2019}.  A linear
$[n,k]_q$ code is MDS precisely when the columns of a generator matrix form
a $k$-dimensional arc \cite[Theorem~17]{BallLavrauw2019}; after
projectivization this is the usual finite-geometric formulation.  We use
the vector-space convention: an \emph{arc} in a $k$-dimensional vector
space over a field $\mathbb F$ is an ordered family of at least $k$
vectors such that every $k$ are linearly independent.  This is the
convention of \cite[p.~1]{Chowdhury2016}; see also
\cite[Section~2]{BallLavrauw2019} for the projective formulation.  The ordering serves only to fix inherited orders and
determinant signs and does not affect the arc property or its size.  For
$2\leqslant k\leqslant q$, the \emph{MDS conjecture} asserts that a
$k$-dimensional arc over $\mathbb F_q$ has at most $q+1$ elements, apart
from the two classical even-characteristic exceptions.  The conjecture is
known over prime fields \cite[Theorem~1.7]{Ball2012}, but remains open in
general over proper extension fields; see \cite{BallLavrauw2019} for a
survey.

Let $q=p^h$, where $p$ is an odd prime and $h\geqslant2$.  The general
unconditional ranges over proper extension fields recalled below are on
the square-root scale.  When $h$ is odd, Voloch proved that the
MDS conjecture holds throughout the range
$k<\frac14\sqrt{pq}-\frac{29}{16}p+4$
\cite[Theorem~3]{Voloch1991}; when $h$ is even, Ball and Lavrauw proved it
throughout the range $k\leqslant\sqrt q-\frac{\sqrt q}p+2$
\cite[Corollary~4]{BallLavrauw2018}.  Further results under additional
hypotheses are summarized in \cite[Section~19]{BallLavrauw2019}.  Our main
theorem replaces these general square-root-scale ranges by a range linear in
$q$.

\begin{samepage}
\begin{theorem}\label{thm:mds-range}
Let $q$ be a power of an odd prime $p$. Define $B(p,q)$ and the range of $k$ by
\begin{equation}\label{intro:main-range}
 B(p,q):=\left\lfloor\frac{(p-2)q+6p-10}{2p-3}\right\rfloor,
 \qquad 2\leqslant k\leqslant B(p,q)\quad\text{or}\quad
 q+2-B(p,q)\leqslant k\leqslant q.
\end{equation}
Then the maximum size of a $k$-dimensional arc over $\mathbb F_q$ is
$q+1$. Equivalently, the MDS conjecture over $\mathbb F_q$ holds in these
dimensions.
\end{theorem}
\end{samepage}

For fixed $p$,
\[
 B(p,q)=\frac{p-2}{2p-3}\,q+O_p(1),
\]
so the first interval in Theorem~\ref{thm:mds-range} is linear in $q$.
The exact number of dimensions left uncovered after adjoining the dual
interval is $\max\{0,q+1-2B(p,q)\}$, which equals
$\frac{q}{2p-3}+O_p(1)$ as $q$ tends to infinity through powers of the
fixed prime $p$.  Thus our theorem proves the conjecture in an asymptotic
proportion $1-1/(2p-3)$ of all dimensions.  For comparison, Ball and
De~Beule proved the conjecture for $k\leqslant2p-2$ over proper extensions
\cite[Theorem~5.2]{BallDeBeule2012}; more generally, for fixed odd $p$,
the general unconditional ranges recalled above cover $O_p(\sqrt q)$
dimensions even after adjoining their duals.  Thus our result changes the
scale from $q^{1/2}$ to $q$, rather than merely enlarging a previous range
by a constant.

The numerical threshold in \eqref{intro:main-range} is not itself new.
Ball had already identified it conditionally: a conjectured weight-one
column-space condition for the relevant matrix would exclude $(q+2)$-point
arcs whenever
$k\leqslant((p-2)q+6p-10)/(2p-3)$; see
\cite[Conjecture~1 and the paragraph following it]{Ball2018Extending}.
A closely related conditional formulation is due to Chowdhury, who
introduced the matrices $M_G^{\uparrow n}$
\cite[Definition~1.3]{Chowdhury2016} and conjectured
\cite[Conjecture~1.9]{Chowdhury2016} full row rank at
$\lvert G\rvert=2k-3+n$ under
\[
 k\leqslant
 \min\left\{p+n(p-2),\frac{q+4-n}2\right\}.
\]
We reach Ball's threshold unconditionally by proving a stronger rank
theorem over arbitrary fields of positive characteristic.  For a
$k$-dimensional arc $G$ of size $m$ over a field $\mathbb F$ of characteristic
$p>0$, Theorem~\ref{gl:full-rank} proves that $M_G^{\uparrow n}$ has full row rank throughout
\[
 k+n-1\leqslant m\leqslant
 k+n-1+\min\{k-2,(n+1)(p-2)\}.
\]
Our full-row-rank theorem requires no finiteness assumption on $\mathbb F$.  At the
endpoint $m=2k-3+n$, its condition becomes
$k\leqslant p+n(p-2)$, proving Chowdhury's full-row-rank conjecture without
the additional $q$-dependent restriction.  Specializing to $\mathbb F_q$
and taking $n=e=q-2k+5$, the inequality
$k\leqslant p+e(p-2)$ is equivalent to $k\leqslant B(p,q)$; with
$r=k-p$, it is $r\leqslant e(p-2)$.  From a hypothetical $(q+2)$-point arc, the construction in the proof of
\cite[Theorem~12.1]{BallLavrauw2019} produces a nonzero left-kernel vector
for a matrix that Theorem~\ref{gl:full-rank} forces to have full row rank.  This gives the
finite-field contradiction needed for the range $k\leqslant B(p,q)$.  Chowdhury had
previously shown that her conjecture implies the MDS conjecture in its
original $(q+1)$-point parameter range
\cite[Corollary~1.10]{Chowdhury2016}.

The distinction between the $(q+1)$- and $(q+2)$-point cases can matter at the
center.  Chowdhury's original parameter range yields the smaller threshold
obtained from $B(p,q)$ by replacing the numerator $(p-2)q+6p-10$ with
$(p-2)q+5p-8$; denote it by $B_{\mathrm{Ch}}(p,q)$.  Direct use of our
arbitrary-field theorem at size $q+2$ gives $B(p,q)$, with
$B(p,q)-B_{\mathrm{Ch}}(p,q)\in\{0,1\}$.  Although the two thresholds differ by at most one, this one-step gain can
close the central pair of dimensions.  For example,
$B_{\mathrm{Ch}}(7,49)=24$ whereas $B(7,49)=25$, so
Theorem~\ref{thm:mds-range} covers dimensions $25$ and $26$ and hence
proves the MDS conjecture over $\mathbb F_{49}$.

Our MDS range also has two consequences in finite geometry and coding
theory.  We state them here and defer the application-specific definitions
to Section~\ref{sec:nrc-prs-applications}.  In Corollary~\ref{cor:prs-covering-radius},
$\operatorname{PRS}_q(k)$ denotes the projective Reed--Solomon code over
$\mathbb F_q$ of length $q+1$ and dimension $k$, and $\rho$ its covering
radius.

\begin{corollary}\label{cor:nrc-complete}
Let $q$ be an odd prime power and let $2\leqslant N\leqslant q-2$.  Every
normal rational curve in $\mathrm{PG}(N,q)$ is a complete $(q+1)$-arc.
\end{corollary}

\begin{corollary}\label{cor:prs-covering-radius}
Let $q$ be an odd prime power and let $2\leqslant k\leqslant q-2$.  Then
\[
 \rho\bigl(\operatorname{PRS}_q(k)\bigr)=q-k.
\]
\end{corollary}

Corollary~\ref{cor:nrc-complete} proves the odd-characteristic case of the
classical \emph{normal-rational-curve completeness conjecture}, whose study goes
back to Segre \cite{Segre1955}; for the modern formulation used here, see
\cite[Conjecture~1.2]{BartoliDavydovMarcuginiPambianco2018}.
Corollary~\ref{cor:prs-covering-radius} proves the odd-characteristic
\emph{projective Reed--Solomon covering-radius conjecture} of Zhang and Wan
\cite[Conjecture~1.3]{ZhangWan2016}.  D{\"u}r identified this covering-radius
problem with completeness of the corresponding normal rational curve
\cite{Dur1994}; see also the modern restatement
\cite[Lemma~II.4]{ZhangWanKaipa2020}.

Our proof separates a positive-characteristic algebraic part from the
finite-field argument.  Section~\ref{sc:section} develops inclusion
operators on subsets and identifies the relevant kernels with images of
iterated inclusion maps in characteristic $p$.  Section~\ref{sec:local}
uses these identities to convert local alternating relations on an arc
into representations indexed by smaller subsets.  Section~\ref{gl:section}
uses this reduction to prove a global vanishing theorem, from which it
deduces the determinant relations and the full-row-rank theorem for
Chowdhury's matrices.  None of
these three sections uses finiteness of the base field.  The finite field
enters only in Section~\ref{sec:mds-deduction}, where the Ball--Lavrauw
construction provides the nonzero left-kernel vector described above;
classical duality of linear MDS codes \cite[p.~117]{Singleton1964} gives
the second interval.  The substitutions
$r=k-p$ and $e=q-2k+5$ identify the inequality
$r\leqslant e(p-2)$ with the threshold $k\leqslant B(p,q)$ in
\eqref{intro:main-range}.  Finally,
Section~\ref{sec:nrc-prs-applications} combines our MDS theorem with known
completeness results for normal rational curves to prove the two corollaries.

\section{Inclusion operators on subsets}\label{sc:section}

This section develops the properties of inclusion maps on subsets needed
later.  We state the main result, Proposition~\ref{sc:structure}, after the
definitions and then prove it.

Throughout this section, let $X$ be a finite set, let $\mathbb F$ be a field of
characteristic $p>0$, and let $H$ be an $\mathbb F$-vector space.  For $a\geqslant0$ put
\[
 \binom Xa:=\{A\subseteq X:\lvert A\rvert=a\},
 \qquad
 F_a(X;H):=\{f\mid f:\tbinom Xa\to H\text{ is a function}\}.
\]
Pointwise addition and scalar multiplication make $F_a(X;H)$ an
$\mathbb F$-vector space.  For $a<0$, set $F_a(X;H):=\{0\}$,
the zero vector space.
For $0\leqslant b\leqslant a$ define
\[
 (E_{a,b}f)(A):=\sum_{B\in\binom Ab}f(B),
 \qquad E_{a,b}:F_b(X;H)\longrightarrow F_a(X;H),
\]
and interpret $E_{a,b}$ as zero when $b<0$.  Counting intermediate
subsets gives
\begin{equation}\label{sc:composition}
 E_{a,b}E_{b,c}=\binom{a-c}{b-c}E_{a,c}
 \qquad\text{for }0\leqslant c\leqslant b\leqslant a.
\end{equation}

Define
\[
 d:\bigoplus_{a\geqslant0}F_a(X;H)\longrightarrow
   \bigoplus_{a\geqslant0}F_a(X;H),
 \qquad d|_{F_a(X;H)}:=E_{a+1,a}.
\]
Repeated use of~\eqref{sc:composition} gives
\[
 d^m|_{F_a(X;H)}=m!\,E_{a+m,a}\qquad(m\geqslant0).
\]
Hence $d^p=0$ in characteristic $p$.  This gives the containments
\[
 \operatorname{im}d^{p-1}\subseteq\ker d,
 \qquad
 \operatorname{im}d\subseteq\ker d^{p-1}.
\]
Proposition~\ref{sc:structure} proves the corresponding equalities in the
specified degrees under explicit bounds on $\lvert X\rvert$.

For $a\geqslant1$, an \emph{$a$-matching} on $X$ is a family
$P=((x_i,y_i))_{i=1}^a$ of pairwise disjoint ordered pairs.  Its
\emph{support} is the set of their $2a$ entries.  Define
\[
 \mathcal T(P):=
 \left\{J\in\binom Xa:
 \lvert J\cap\{x_i,y_i\}\rvert=1\text{ for }1\leqslant i\leqslant a\right\}.
\]
For $J\in\mathcal T(P)$, define
\[
 \epsilon_P(J):=(-1)^{\lvert \{i:y_i\in J\}\rvert}.
\]
Changing the choice in one pair reverses the sign.  For
$f\in F_a(X;H)$ put
\[
 (\Delta_af)(P):=
 \sum_{J\in\mathcal T(P)}\epsilon_P(J)f(J).
\]
For $a=0$, take $P$ to be the empty matching and let $\Delta_0$ be
evaluation at the empty set.  Whenever
$\lvert X\rvert\geqslant2a$, define
\[
 \mathcal L_a(X;H):=
 \{f\in F_a(X;H):(\Delta_af)(P)=0
   \text{ for every $a$-matching }P\}.
\]

\begin{proposition}\label{sc:structure}
Let $a\geqslant0$.  In the following maps, $F_j$ denotes $F_j(X;H)$.
\begin{enumerate}[label=(\roman*),ref=\roman*]
\item\label{item:sc-kernel-d} If $\lvert X\rvert\geqslant2a+1$, then
$\ker E_{a+1,a}=\operatorname{im}E_{a,a+1-p}$.  Equivalently,
\[
 \ker\bigl(d:F_a\to F_{a+1}\bigr)
 =\operatorname{im}\bigl(d^{p-1}:F_{a-p+1}\to F_a\bigr).
\]
\item\label{item:sc-kernel-dp} If $\lvert X\rvert\geqslant2a+p-1$, then
$\ker E_{a+p-1,a}=\operatorname{im}E_{a,a-1}$.  Equivalently,
\[
 \ker\bigl(d^{p-1}:F_a\to F_{a+p-1}\bigr)
 =\operatorname{im}\bigl(d:F_{a-1}\to F_a\bigr).
\]
\item\label{item:sc-low} If $\lvert X\rvert\geqslant2a$, then
\begin{align}
 \mathcal L_a(X;H)
 &=\sum_{0\leqslant b<a}\operatorname{im}E_{a,b},
 \label{sc:all-low}\\
 &=\operatorname{im}E_{a,a-1}+
 \sum_{m=1}^{\lfloor a/p\rfloor}
 \operatorname{im}E_{a,a-mp}.
 \label{sc:lower-images}
\end{align}
The sums are zero when $a=0$.
\end{enumerate}
In particular, $d:F_a\to F_{a+1}$ is injective for $0\leqslant a<p-1$
whenever $\lvert X\rvert\geqslant2a+1$.
\end{proposition}

Parts~(\ref{item:sc-kernel-d})--(\ref{item:sc-kernel-dp}) identify the
kernels of $d$ and $d^{p-1}$, while part~(\ref{item:sc-low}) describes the
functions for which all the alternating sums $(\Delta_af)(P)$ vanish.
Related homological results for inclusion maps appear in
\cite{Bier1993Homological,MnukhinSiemons1996}.  We give a direct proof in
the present notation, with the cardinality hypotheses made explicit.

For the proof, put $N:=\lvert X\rvert$, order $X$, and identify it with
$[N]$.  For
$I=\{i_1<\cdots<i_j\}\subseteq[N]$ satisfying
$i_s\geqslant2s$ for every $s$, let $e_I$ denote the unit vector at $I$
and set
\[
 \langle I\rangle_a:=E_{a,j}e_I,
 \qquad \lvert I\rvert=j\leqslant a.
\]
Then $\langle I\rangle_a(A)=\mathds{1}_{I\subseteq A}$, where the
indicator is $1$ when $I\subseteq A$ and $0$ otherwise.

\begin{lemma}\label{sc:subset-basis}
For $0\leqslant a\leqslant N$, the vectors $\langle I\rangle_a$ indexed by
sets $I=\{i_1<\cdots<i_j\}\subseteq[N]$ satisfying
$i_s\geqslant2s$ for every $s$ and
$j\leqslant\min\{a,N-a\}$ form a basis of the free abelian group
$\mathbb Z^{\binom{[N]}a}$.
\end{lemma}

\begin{proof}
Plaza recalls Frankl's rank of a subset in
\cite[Definition~66]{Plaza2016} and proves the corresponding basis theorem
in \cite[Corollary~71]{Plaza2016}.  For
$I=\{i_1<\cdots<i_j\}$, rank $j$ is equivalent to $i_s\geqslant2s$ for
all $s$, so the cited result gives exactly the basis above.  Its transition
matrix has integer entries and is nonsingular modulo every prime; hence its
determinant is $\pm1$.
\end{proof}

If $0\leqslant a\leqslant b$ and $N\geqslant a+b$, then for every
indexing set $I$ from the lemma with $\lvert I\rvert\leqslant a$,
\begin{equation}\label{sc:diagonal}
 E_{b,a}\langle I\rangle_a
 =\binom{b-i}{a-i}\langle I\rangle_b,
 \qquad i=\lvert I\rvert.
\end{equation}
\begin{proof}[Proof of Proposition~\ref{sc:structure}]
We first work with scalar coefficients.  The two image
containments in parts~(\ref{item:sc-kernel-d})--(\ref{item:sc-kernel-dp})
follow from $d^p=0$.

For~(\ref{item:sc-kernel-d}), equation~\eqref{sc:diagonal} gives
$d\langle I\rangle_a=(a+1-i)\langle I\rangle_{a+1}$, where
$i:=\lvert I\rvert$.  Hence the kernel is spanned by the basis vectors
with $i=a+1-mp\geqslant0$, $m\geqslant1$.  Each such $I$ also indexes
a basis vector in $F_{a+1-p}$.  The map $d^{p-1}:F_{a+1-p}\to F_a$
equals $(p-1)!E_{a,a+1-p}$, and
\[
 \binom{mp-1}{p-1}\equiv1\pmod p.
\]
Thus $d^{p-1}$ maps onto $\ker d$, proving
part~(\ref{item:sc-kernel-d}) and the stated injectivity consequence.

For~(\ref{item:sc-kernel-dp}), equation~\eqref{sc:diagonal} gives
\[
 \binom{a-i+p-1}{p-1}\equiv
 \begin{cases}
 1,&p\mid(a-i),\\
 0,&p\nmid(a-i)
 \end{cases}
 \pmod p.
\]
Hence the kernel of $d^{p-1}$ is spanned by the basis vectors
$\langle I\rangle_a$ with $p\nmid(a-i)$, where $i=\lvert I\rvert$.
These are exactly the basis vectors in the image of $d:F_{a-1}\to F_a$,
since $d\langle I\rangle_{a-1}=(a-i)\langle I\rangle_a$ for $i<a$.
This proves part~(\ref{item:sc-kernel-dp}).

For~(\ref{item:sc-low}), every vector $E_{a,b}e_I$ with $b<a$ satisfies
$(\Delta_aE_{a,b}e_I)(P)=0$ for every $a$-matching $P$: either a pair is
missed, in which case its two choices cancel, or the value is zero for every
$J\in\mathcal T(P)$.  Hence the right side of~\eqref{sc:all-low} is contained in
$\mathcal L_a$.

For the reverse inclusion, order lexicographically the $a$-sets
$I=\{i_1<\cdots<i_a\}$ satisfying $i_s\geqslant2s$.
Choose distinct $j_s<i_s$ outside $I$: at least
$i_s-s\geqslant s$ elements below $i_s$ lie outside $I$, and at most
$s-1$ were chosen earlier.  Consider the matching
$P=((j_1,i_1),\ldots,(j_a,i_a))$.  Every
$J\in\mathcal T(P)$ other than $I$ is lexicographically smaller.  The functionals
$f\mapsto(\Delta_af)(P)$ obtained from these matchings form, on the basis
vectors $\langle I\rangle_a$ with $\lvert I\rvert=a$, a triangular
matrix with diagonal $(-1)^a$; they vanish on all basis vectors indexed by
sets of smaller cardinality.  Hence no basis vector indexed by an $a$-set
occurs in an element of $\mathcal L_a$, proving~\eqref{sc:all-low}.

Finally,
$E_{a,a-1}\langle I\rangle_{a-1}=(a-i)\langle I\rangle_a$ for
$i=\lvert I\rvert<a$.  If $p\nmid a-i$, then
$\langle I\rangle_a\in\operatorname{im}E_{a,a-1}$; if
$p\mid a-i$, then $i=a-mp$ and
$\langle I\rangle_a=E_{a,a-mp}e_I$.  This gives
\eqref{sc:lower-images}.  The argument applies coordinatewise to
$H$-valued functions.
\end{proof}

The sum in~\eqref{sc:lower-images} need not be direct, and the
representation need not be unique.  The additional indices $i=a-mp$
are precisely those with $i<a$ for which $a-i=0$ in $\mathbb F$.

In addition to the main result of this section,
Proposition~\ref{sc:structure}, Section~\ref{sec:local} also uses
Equation~\eqref{sc:composition} and Lemma~\ref{sc:subset-basis}.

\section{Matching relations on arcs}\label{sec:local}

Throughout this section, let $\mathbb F$ have characteristic $p>0$, let
$X$ be a finite ordered set, let $V$ be a $D$-dimensional
$\mathbb F$-vector space, and let $(w_x)_{x\in X}$ be an arc in $V$.
Put $N:=\lvert X\rvert$ and fix $1\leqslant e\leqslant D$.  Then
$N\geqslant D$, and the arc spans $V$.  The order on $X$ serves only to
provide inherited orders for indexed subsets and products; it plays no role
in linear independence.  For $A\subseteq X$, write
$W_A:=\operatorname{span}\{w_a:a\in A\}$.  If $\mathbf x$ is an ordered
tuple of distinct indices, write $\underline{\mathbf x}$ for its underlying
set; if $z\notin\underline{\mathbf x}$, write $\mathbf xz$ for the tuple
obtained by appending $z$ to $\mathbf x$.

Only Theorem~\ref{loc:slr} is used in later sections.  It converts vector-valued
matching relations into coefficient functions generated from smaller
subsets.  We state it first and then introduce the coefficient families
needed for its proof.

\begin{theorem}\label{loc:slr}
Let $t\geqslant1$ and $1\leqslant e\leqslant D$, and suppose
\[
 N\geqslant2t+e-1,\qquad
 D\geqslant t+e-1+\max\{0,t-(e-1)(p-2)\}.
\]
For each $U\in\binom Xe$, let
$h_U:\binom{X\setminus U}{t-1}\longrightarrow\mathbb F$ be a function.  Suppose that for every ordered $(e-1)$-tuple $\mathbf a$ of distinct
elements of $X$ and every
$t$-matching $\mathcal P$ disjoint from
$\underline{\mathbf a}$,
\begin{equation}\label{loc:geometric-equation}
 \sum_{J\in\mathcal T(\mathcal P)}\epsilon_{\mathcal P}(J)
 \sum_{z\in J}
 h_{\underline{\mathbf a}\cup\{z\}}(J\setminus\{z\})w_z
 \in W_{\underline{\mathbf a}}.
\end{equation}
Then for every $U\in\binom Xe$,
\begin{equation}\label{loc:slr-conclusion}
 h_U\in\mathcal L_{t-1}(X\setminus U;\mathbb F)
 =\sum_{0\leqslant i<t-1}\operatorname{im}E_{t-1,i}.
\end{equation}
In particular, if $t=1$, then $h_U(\varnothing)=0$ for every
$U\in\binom Xe$.
\end{theorem}

The proof uses two gradings.  The \emph{tuple degree} $\mu$ is the
length of the ordered tuple, while the \emph{subset degree} $b$ is the size
of the coefficient index on which the inclusion operators of
Section~\ref{sc:section} act.  A map $\partial$ moves one index from the
tuple to the subset, preserving $\mu+b$; the operator $d$ raises only the
subset degree.  The two maps commute.

The proof has three ingredients.  Linear independence gives
$\partial$-exactness in positive tuple degree.  At tuple degree zero,
kernel terms for $d$ or $d^{p-1}$ impose the required quotient condition
without changing the prescribed inclusion image.  Proposition~\ref{sc:structure}~(\ref{item:sc-kernel-d})--(\ref{item:sc-kernel-dp}),
together with the exactness of the $\partial$-maps, then allows the
relations to be lifted through decreasing tuple degree; a final descent
removes the remaining terms in decreasing subset degree.

For $m\geqslant1$, let $\Lambda^mV$ denote the \emph{exterior power} of $V$: it is generated by symbols
$v_1\wedge\cdots\wedge v_m$, linear in each factor and zero when two
factors coincide.  Put $\Lambda^0V:=\mathbb F$, and let $\wedge$ denote
the product induced by concatenating pure symbols.

\subsection{Gradings and the boundary operator}

The ordered tuple records the exterior factors already present in a
coefficient, while the disjoint subset is the index on which the inclusion
operators act.  The boundary operator moves one index from the tuple to
the subset without changing their union.

For $0\leqslant\mu\leqslant e$ and an ordered $\mu$-tuple
$\mathbf x=(x_1,\ldots,x_\mu)$ of distinct indices, put
$w_{\mathbf x}:=w_{x_1}\wedge\cdots\wedge w_{x_\mu}$ and
\[
 A_{\mathbf x}^{(e)}:=w_{\mathbf x}\wedge\Lambda^{e-\mu}V
 \subseteq\Lambda^eV.
\]
For the empty tuple, $A_\varnothing^{(e)}=\Lambda^eV$.  Reordering
$\mathbf x$ by a permutation $\sigma$ multiplies $w_{\mathbf x}$ by
$\operatorname{sgn}(\sigma)$, viewed as a scalar in $\mathbb F$; in
characteristic $2$ this scalar is always $1$.  In every characteristic,
$A_{\mathbf x}^{(e)}$ depends only on $\underline{\mathbf x}$.

Fix a subset $S\subseteq X$, regarded as the union of the tuple and subset
indices.  The subset index is then the complement in $S$ of the tuple
entries, so only the alternating tuple index remains.  The following spaces
and maps require no independence assumption.  For
$0\leqslant\mu\leqslant e$, let
$\mathcal C_e^\mu(S)$ be the space of
alternating families
\[
 \omega=(\omega_{\mathbf x}),\qquad
 \omega_{\mathbf x}\in A_{\mathbf x}^{(e)},
\]
indexed by ordered $\mu$-tuples of distinct elements of $S$.
Alternation means
\[
 \omega_{(x_{\sigma(1)},\ldots,x_{\sigma(\mu)})}
 =\operatorname{sgn}(\sigma)\omega_{(x_1,\ldots,x_\mu)}
 \qquad\text{for all }\sigma\in\mathfrak S_\mu.
\]
In characteristic $2$, this condition becomes permutation-invariance
because $\operatorname{sgn}(\sigma)=1$ in $\mathbb F$.  This is consistent
with the preceding alternating convention: repeated exterior factors still
vanish.
For $\mu=0$ this space is $\Lambda^eV$; for $\mu<0$ or $\mu>e$ it is
zero.  For $1\leqslant\mu\leqslant e$, define
$\partial_S:\mathcal C_e^\mu(S)\to\mathcal C_e^{\mu-1}(S)$ by
\[
 (\partial_S\omega)_{\mathbf x}
 :=\sum_{z\in S\setminus\underline{\mathbf x}}
 \omega_{\mathbf x z}.
\]
Here $\mathbf x$ has length $\mu-1$.  At tuple degree zero, set
$\partial_S:=0$;
this is distinct from the
quotient map $\pi_S$ below.  Alternation gives $\partial_S^2=0$.  Let
$\pi_S:\Lambda^eV\to\Lambda^e(V/W_S)$ be the natural quotient map, and put
\[
 \mathcal I_S:=W_S\wedge\Lambda^{e-1}V=\ker\pi_S.
\]
The kernel identity follows by extending a basis of $W_S$ to a basis of
$V$: exactly the exterior coordinates with at least one factor in $W_S$
are killed by the quotient.  In particular, $\mathcal I_\varnothing=0$.

\begin{lemma}\label{ex:exact}
For every nonempty $S\subseteq X$ such that $(w_s)_{s\in S}$ is linearly independent, put $m:=\min\{e,\lvert S\rvert\}$.
Then the sequence
\[
 0\longrightarrow \mathcal C_e^{m}(S)
 \xrightarrow{\partial_S}\cdots
 \xrightarrow{\partial_S}\mathcal C_e^1(S)
 \xrightarrow{\partial_S}\Lambda^eV
 \xrightarrow{\pi_S}\Lambda^e(V/W_S)\longrightarrow0
\]
is exact.  For $S=\varnothing$, the assertion is the exact sequence
$0\to\Lambda^eV\xrightarrow{\mathrm{id}}\Lambda^eV\to0$.  Here
``exact'' means that the image of each map is the kernel of the next.
\end{lemma}

\begin{proof}
Extend $\{w_s:s\in S\}$ to a basis $(v_b)_{b\in\mathcal B}$ of $V$.
Fix an exterior basis vector $v_K$, $\lvert K\rvert=e$, and put $T=K\cap S$.
The $v_K$-coordinate can occur in the component indexed by
$\mathbf x$ exactly when $\underline{\mathbf x}\subseteq T$.  On this
coordinate, $\partial_S$ is the alternating deletion map on ordered
subsets of $T$.  If $T\neq\varnothing$, fix $t_0\in T$.  On scalar
coordinate arrays, define
$(hf)_{\mathbf x t_0}=f_{\mathbf x}$ when $t_0$ does not occur in
$\mathbf x$, extend by alternation, and set $hf$ to zero on tuples
not containing $t_0$.  Separating the summand indexed by $t_0$ gives
\[
 \partial h+h\partial=\operatorname{id}.
\]
If $T=\varnothing$, only degree $0$ remains and $v_K$ survives precisely
in $\Lambda^e(V/W_S)$.  Taking the direct sum over all $K$ proves exactness.
\end{proof}

If $(w_s)_{s\in S}$ is linearly independent, the lemma makes every
$\partial$-closed array of positive tuple degree $\partial$-exact.  At
tuple degree zero, the boundary image is
$\mathcal I_S=\ker\pi_S$; this identity holds for every $S$, because the
image is the sum of the spaces $w_s\wedge\Lambda^{e-1}V$.

We now take the direct sum of these complexes over all subsets $S$;
positive-degree exactness is used only on the corresponding summands.
For $0\leqslant\mu\leqslant e$ and $b\geqslant0$, set
\begin{equation}\label{ex:global-C}
 C_e(\mu,b):=
 \bigoplus_{S\in\binom X{\mu+b}}\mathcal C_e^\mu(S).
\end{equation}
Equivalently, an element $F\in C_e(\mu,b)$ is a family
$F_{\mathbf x}(I)\in A_{\mathbf x}^{(e)}$, where $\mathbf x$ is an ordered
$\mu$-tuple, $I$ is a $b$-subset disjoint from
$\underline{\mathbf x}$, and the family is alternating in
$\mathbf x$.  The operator $\partial$ preserves the union
$S=\underline{\mathbf x}\sqcup I$.
Put $C_e(\mu,b):=0$ if $b<0$, $\mu<0$, or $\mu>e$.

For $\mu\geqslant1$, the direct sum of the $\partial_S$ gives
\[
 \partial:C_e(\mu,b)\longrightarrow C_e(\mu-1,b+1),
 \qquad
 (\partial F)_{\mathbf x}(J)
 :=\sum_{z\in J}F_{\mathbf x z}(J\setminus\{z\}).
\]
For $\mu=0$, set $\partial:=0$.  For $0\leqslant b\leqslant a$, the
inclusion operator acts on the
subset index only:
\[
 (E_{a,b}F)_{\mathbf x}(J)
 =\sum_{I\in\binom Jb}F_{\mathbf x}(I).
\]
Consequently
\begin{equation}\label{ex:commutation}
 \partial^2=0,
 \qquad
 \partial E_{a,b}=E_{a+1,b+1}\partial.
\end{equation}

Let $d$ denote the inclusion operator on the subset index,
$d|_{C_e(\mu,b)}:=E_{b+1,b}$.  Then
\[
 d^p=0,
 \qquad
 \partial d=d\partial.
\]
The two moves are
\[
 \partial:(\mu,b)\longmapsto(\mu-1,b+1),
 \qquad d:(\mu,b)\longmapsto(\mu,b+1).
\]
Thus the following diagram commutes:
\[
\begin{tikzcd}[column sep=2.6em,row sep=2.2em]
 C_e(\mu+1,b-2) \arrow[r,"\partial"] \arrow[d,"d"'] &
 C_e(\mu,b-1) \arrow[r,"\partial"] \arrow[d,"d"'] &
 C_e(\mu-1,b) \arrow[d,"d"] \\
 C_e(\mu+1,b-1) \arrow[r,"\partial"] \arrow[d,"d"'] &
 C_e(\mu,b) \arrow[r,"\partial"] \arrow[d,"d"'] &
 C_e(\mu-1,b+1) \arrow[d,"d"] \\
 C_e(\mu+1,b) \arrow[r,"\partial"'] &
 C_e(\mu,b+1) \arrow[r,"\partial"'] &
 C_e(\mu-1,b+2)
\end{tikzcd}
\]
Here $\partial$ preserves $\mu+b$, whereas $d$ increases it.  When
$\mu+b\leqslant D$, the arc property makes the vectors on every fixed set
$S$ linearly independent, so Lemma~\ref{ex:exact} applies to the
$\partial$-maps on that set.  With a tuple of length
$\mu$ fixed, Proposition~\ref{sc:structure}~(\ref{item:sc-kernel-d})--(\ref{item:sc-kernel-dp})
identifies the kernels of $d$ and $d^{p-1}$ on the $N-\mu$ remaining
indices, under their respective cardinality hypotheses.  Neither equality
follows from $d^p=0$ alone.

\subsection{The quotient condition at tuple degree zero}

Exactness controls positive tuple degree.  At tuple degree zero, however,
being a boundary means satisfying an additional quotient condition: the
value at a subset $I$ must lie in $\mathcal I_I$.  Accordingly, for $a\geqslant0$, set
\[
 \mathcal S_a(X):=\{c\in F_a(X;\Lambda^eV):
          c(I)\in\mathcal I_I\text{ for every }I\in\tbinom Xa\}.
\]
Equivalently, $c\in\mathcal S_a(X)$ exactly when every $c(I)$ maps to
zero in $\Lambda^e(V/W_I)$.  Since
$\mathcal I_I\subseteq\mathcal I_J$ when $I\subseteq J$, the inclusion
operators preserve this condition.  Lemma~\ref{ex:exact} and
Proposition~\ref{sc:structure}~(\ref{item:sc-kernel-d})--(\ref{item:sc-kernel-dp})
alone do not impose this quotient condition.  The next lemma enforces it by adding a kernel term
without changing the prescribed inclusion image; the following lemma treats
a sum of images generated from smaller subsets.

\begin{lemma}\label{sp:repair}
Let $a\geqslant0$ and $c\in F_a(X;\Lambda^eV)$.
\begin{enumerate}[label=(\arabic*),ref=\arabic*]
\item\label{item:sp-repair-d} If $D\geqslant2a+e+1$ and
$(E_{a+1,a}c)(J)\in\mathcal I_J$ for every
$J\in\binom X{a+1}$, then there exists
$q\in F_{a+1-p}(X;\Lambda^eV)$ such that
$c-E_{a,a+1-p}q\in\mathcal S_a(X)$.
\item\label{item:sp-repair-dp} If $D\geqslant2a+p+e-1$ and
$(E_{a+p-1,a}c)(J)\in\mathcal I_J$ for every
$J\in\binom X{a+p-1}$, then there exists
$q\in F_{a-1}(X;\Lambda^eV)$ such that
$c-E_{a,a-1}q\in\mathcal S_a(X)$.
\end{enumerate}
In both cases the correction is in the kernel of the operator in the
hypothesis, so the image under that operator is unchanged.
\end{lemma}

\begin{proof}
By~\eqref{sc:composition}, both corrections lie in the required
kernels: each relevant composite is $p$ times an inclusion operator.  We
prove the two assertions simultaneously by induction on $a$.

For $a=0$, in case~(\ref{item:sp-repair-d}) the image of $c(\varnothing)$ vanishes in the
quotient by every $W_{\{x\}}$, so
$c(\varnothing)\in\bigcap_x\mathcal I_{\{x\}}=0$; it suffices to use
$e+1$ independent arc vectors.  In case~(\ref{item:sp-repair-dp}), choose $B\subseteq X$ so that $(w_b)_{b\in B}$ is a basis of $V$.
For each $K\in\binom Be$, the bound $D-e\geqslant p-1$ supplies a
$(p-1)$-set $J\subseteq B\setminus K$.  The $K$-coordinate of
$c(\varnothing)$ is zero because its image in
$\Lambda^e(V/W_J)$ is zero.  Thus again $c(\varnothing)=0$.

Let $a\geqslant1$ and choose $B\subseteq X$ so that $(w_b)_{b\in B}$ is a basis of $V$.  For
$K\in\binom Be$, let $w_K$ denote the wedge of the vectors indexed by
$K$ in inherited order.  The image of $w_K$ in $\Lambda^e(V/W_I)$, for $I\subseteq B$,
is nonzero exactly when $I\cap K=\varnothing$.  For fixed $K$, the
$w_K$-coordinates of $c(I)$ on these sets form a function on the
$a$-subsets of $B\setminus K$.  In case~(\ref{item:sp-repair-d}) the
hypothesis puts this function in $\ker d$; in
case~(\ref{item:sp-repair-dp}) it lies in $\ker d^{p-1}$.  Since $\lvert B\setminus K\rvert=D-e$, the corresponding cardinality
hypothesis of Proposition~\ref{sc:structure} holds.  Its parts
(\ref{item:sc-kernel-d}) and~(\ref{item:sc-kernel-dp}) supply a preimage
under $d^{p-1}$ in the first case and under $d$ in the second.  Absorb the nonzero scalar $(p-1)!$ to write the
first correction with $E_{a,a+1-p}$.  Extend the coordinate preimages by
zero to $X$, multiply by $w_K$, and sum over $K$.  After subtracting this
kernel correction for the relevant inclusion map,
$c(I)\in\mathcal I_I$ for every $a$-subset $I\subseteq B$.

Assume $c(I)\in\mathcal I_I$ for every $a$-subset $I\subseteq Y$, where
$Y\supseteq B$, and choose $x\notin Y$.  Set
$\overline V=V/\langle w_x\rangle$ and
\[
 \overline c_x(I)=\overline{c(I\cup\{x\})},
 \qquad I\in\binom Y{a-1}.
\]
The projected vectors indexed by $Y$ form a spanning arc in
$\overline V$: adjoining $w_x$ proves independence, while the image of the
basis indexed by $B\subseteq Y$ spans.  Evaluating the original condition
on a target set containing $x$ and projecting shows that $\overline c_x$
satisfies the same quotient condition in dimension $D-1$ and subset degree
$a-1$.  The relevant inclusion map is $d$ in
case~(\ref{item:sp-repair-d}) and $d^{p-1}$ in
case~(\ref{item:sp-repair-dp}); the original dimension bound is stronger
than the inductive one by one.

To lift the inductive correction, put $r=a+1-p$ in
case~(\ref{item:sp-repair-d}) and $r=a-1$ in
case~(\ref{item:sp-repair-dp}).  Any nonzero inductive correction has
$r\geqslant1$ and is $E_{a-1,r-1}\overline q$ for some
$\overline q\in F_{r-1}(Y;\Lambda^e\overline V)$.  Choose a linear section
$\sigma_x:\Lambda^e\overline V\to\Lambda^eV$ and set
\[
 q_x(L)=
 \begin{cases}
 \sigma_x\bigl(\overline q(L\setminus\{x\})\bigr),
 &x\in L,\ L\setminus\{x\}\subseteq Y,\\
 0,&\text{otherwise}.
 \end{cases}
\]
Then, for $I\in\binom Y{a-1}$ and $J\in\binom Ya$,
\[
 \overline{(E_{a,r}q_x)(I\cup\{x\})}
 =(E_{a-1,r-1}\overline q)(I),
 \qquad (E_{a,r}q_x)(J)=0.
\]
Hence $E_{a,r}q_x$ realizes the quotient correction on the new indices
and vanishes on the old ones.  It lies in the kernel of the relevant
inclusion operator because~\eqref{sc:composition} gives
\[
 E_{a+1,a}E_{a,a+1-p}=pE_{a+1,a+1-p}=0,
 \qquad
 E_{a+p-1,a}E_{a,a-1}=pE_{a+p-1,a-1}=0.
\]
The quotient of $\overline V$ by the projected span indexed by $I$ is
naturally isomorphic to $V/W_{I\cup\{x\}}$.  Any ambiguity in the
lift lies in
$w_x\wedge\Lambda^{e-1}V\subseteq\mathcal I_{I\cup\{x\}}$.  The corrected new value therefore lies in $\mathcal I_{I\cup\{x\}}$.  If the inductive
correction is zero, take no lift.  Adding the elements of $X\setminus B$
one at a time proves both claims.
\end{proof}

A representation by coefficients on smaller subsets can also be chosen
with each coefficient array in the corresponding $\mathcal S_a(X)$.  The
next lemma records the precise statement.

\begin{lemma}\label{sp:supported-low}
Let $t\geqslant1$ and suppose $D\geqslant2t+e-1$.  Then
\begin{equation}\label{sp:strictness}
 \mathcal S_t(X)\cap\mathcal L_t(X;\Lambda^eV)
 =\sum_{0\leqslant a<t}E_{t,a}\mathcal S_a(X).
\end{equation}
Equivalently, if
$y\in\mathcal S_t(X)\cap\mathcal L_t(X;\Lambda^eV)$, then there exist
$c_a\in\mathcal S_a(X)$ such that
\[
 y=\sum_{0\leqslant a<t}E_{t,a}c_a.
\]
In particular $c_0=0$, since $\mathcal S_0(X)=0$.
\end{lemma}

\begin{proof}
The right side of~\eqref{sp:strictness} lies in $\mathcal S_t(X)$ and in
$\mathcal L_t$, so only the reverse inclusion requires proof.
Proposition~\ref{sc:structure}~(\ref{item:sc-low}) gives an initial decomposition
$y=\sum_{a<t}E_{t,a}c_a$.  We alter its coefficient tuple inside the kernel of the total inclusion map
\[
 \Phi_t:\bigoplus_{a<t}F_a(X;\Lambda^eV)\longrightarrow
 F_t(X;\Lambda^eV),
 \qquad
 (c_a)_a\longmapsto\sum_{a<t}E_{t,a}c_a.
\]

The following elementary relations lie in $\ker\Phi_t$.  If
$0\leqslant a<b<t$ and $Q\in F_a(X;\Lambda^eV)$, then changing
\begin{equation}\label{sp:elementary-relation}
 c_b\mapsto c_b+E_{b,a}Q,
 \qquad
 c_a\mapsto c_a-\binom{t-a}{b-a}Q
\end{equation}
leaves $\Phi_t(c)$ unchanged by~\eqref{sc:composition}.  Conversely, if the ground set has at least $2t-1$ elements, these
elementary relations generate $\ker\Phi_t$.  Using
Lemma~\ref{sc:subset-basis}, descend in $b$.
Relation~\eqref{sp:elementary-relation} moves every basis component
$\langle I\rangle_b$ with $\lvert I\rvert=i<b$ from $c_b$ to the
coefficient $c_i$.  What remains in $c_b$ is a combination of
$\langle I\rangle_b$ with $\lvert I\rvert=b$, and $E_{t,b}$ maps these
to distinct basis vectors $\langle I\rangle_t$, with coefficients in
$\Lambda^eV$.  Hence all such
remainders vanish for a tuple in $\ker\Phi_t$.  No division is used, and an
elementary relation extends to a larger ground set by extending $Q$ by zero.

We prove by induction on $t$ the stronger statement that every
representation of an element of $\mathcal S_t(X)$ can be adjusted by a
finite sum of the elementary relations
\eqref{sp:elementary-relation} so that each coefficient lies in the
corresponding $\mathcal S_a(X)$.  For $t=1$, the only coefficient is
$c_0(\varnothing)$, and $y=E_{1,0}c_0\in\mathcal S_1(X)$ places this
vector in every $\mathcal I_{\{x\}}$; since $D\geqslant e+1$, it is zero.

Let $t\geqslant2$ and choose $B\subseteq X$ so that $(w_b)_{b\in B}$ is a basis of $V$.  Fix
$K\in\binom Be$ and restrict the $K$-coordinates of the coefficient tuple
to $B\setminus K$.  Because $y\in\mathcal S_t(X)$, the total degree-$t$ image of these
coordinates is zero on $B\setminus K$.  Since
$\lvert B\setminus K\rvert=D-e\geqslant2t-1$, the preceding description of
$\ker\Phi_t$ gives elementary relations that eliminate the $w_K$-coordinates
with $I\cap K=\varnothing$.  Extend those relations by zero to $X$ and sum
over $K$.  After these corrections, $c_a(I)\in\mathcal I_I$ whenever $I\subseteq B$.

Suppose $c_a(I)\in\mathcal I_I$ for every coefficient index $I\subseteq Y$,
where $Y\supseteq B$, and choose $x\notin Y$.  In
$\overline V=V/\langle w_x\rangle$, define for $1\leqslant a<t$
\[
 \overline c^{\,x}_{a-1}(I)
 =\overline{c_a(I\cup\{x\})}.
\]
As above, the projected vectors indexed by $Y$ form a spanning arc in
$\overline V$.  Evaluate the total image on $J\cup\{x\}$ with
$J\in\binom Y{t-1}$ and project to $\overline V$.  The resulting quotient
tuple satisfies the same condition in degree $t-1$, so the induction
hypothesis applies.

Consider one quotient elementary relation, with
$0\leqslant\alpha<\beta<t-1$ and
$\overline Q\in F_\alpha(Y;\Lambda^e\overline V)$:
\[
 \overline c^{\,x}_\beta\mapsto
 \overline c^{\,x}_\beta+E_{\beta,\alpha}\overline Q,
 \qquad
 \overline c^{\,x}_\alpha\mapsto
 \overline c^{\,x}_\alpha-
 \binom{t-1-\alpha}{\beta-\alpha}\overline Q.
\]
Choose a linear section $\sigma_x:\Lambda^e\overline V\to\Lambda^eV$ and
let $Q_x\in F_{\alpha+1}(X;\Lambda^eV)$ be supported on sets containing
$x$, with
$Q_x(L)=\sigma_x(\overline Q(L\setminus\{x\}))$ there.  Its lift is
\[
 c_{\beta+1}\mapsto c_{\beta+1}+E_{\beta+1,\alpha+1}Q_x,
 \qquad
 c_{\alpha+1}\mapsto c_{\alpha+1}-
 \binom{t-1-\alpha}{\beta-\alpha}Q_x.
\]
Because
$\binom{t-(\alpha+1)}{(\beta+1)-(\alpha+1)}
 =\binom{t-1-\alpha}{\beta-\alpha}$,
this is exactly~\eqref{sp:elementary-relation} with both degrees increased
by one; in particular, it preserves $\Phi_t(c)$.  Its support makes it
vanish on old indices, while projection recovers the quotient relation.
The quotient by the projected span indexed by $I$ is naturally isomorphic
to $V/W_{I\cup\{x\}}$.
Lifting all quotient relations in this way and repeating over
$X\setminus B$ yields $c_a\in\mathcal S_a(X)$ for every $a<t$, proving
\eqref{sp:strictness}.
\end{proof}

\subsection{Lifting inclusion-map relations}

The next lemma starts from $\partial G=dR$ and descends in tuple degree by
using parts~(\ref{item:sc-kernel-d}) and
(\ref{item:sc-kernel-dp}) alternately.
If the descent reaches tuple degree zero, the preceding subsection imposes
the quotient condition; Lemma~\ref{ex:exact} then lifts the corrected
relation back.  We verify the required bounds during the descent and then lift back to $G$.

\begin{lemma}\label{ex:gapone}
Let $1\leqslant\mu\leqslant e$ and $t\geqslant1$.  Suppose the following
conditions hold:
\begin{align}
 N&\geqslant2t+\mu-1,\qquad D\geqslant t+\mu-1;
 \label{ex:bounds}\\
 D&\geqslant2t-(\mu-1)(p-2)+e-1,
 \qquad\text{whenever }t\geqslant
 1+\left\lceil\frac{\mu-1}{2}\right\rceil(p-2).
 \label{ex:bottom-bound}
\end{align}
If $G\in C_e(\mu,t-1)$ and
$\partial G=dR$ for some $R\in C_e(\mu-1,t-1)$, then there exist
$C\in C_e(\mu,t-2)$ and $Z\in C_e(\mu+1,t-2)$ such that
\[
 G=dC+\partial Z.
\]
In particular, $G=0$ when $t=1$, and $Z=0$ when $\mu=e$.
\end{lemma}

\begin{proof}
\smallskip\noindent\emph{Descending in tuple degree.}
Set $X_{-1}=G$ and $X_0=R$, so that
$\partial X_{-1}=dX_0$.  If
$\partial X_{j-1}=dX_j$, then $d\partial X_j=0$.
Part~(\ref{item:sc-kernel-d}), applied with the remaining tuple fixed,
gives either $\partial X_j=0$ in the injective range or
$\partial X_j=d^{p-1}X_{j+1}$.  If instead
$\partial X_{j-1}=d^{p-1}X_j$, then
$d^{p-1}\partial X_j=0$, and
part~(\ref{item:sc-kernel-dp}) gives
$\partial X_j=dX_{j+1}$.
At each application of
part~(\ref{item:sc-kernel-d}) or~(\ref{item:sc-kernel-dp}), choose a
preimage for one orientation of each underlying remaining tuple and extend to the other
orientations by the alternating sign rule.  Since $d$ and $d^{p-1}$ act
only on the subset index, the resulting $X_{j+1}$ lies in the required
space $C_e(\cdot,\cdot)$.  Thus the maps on the subset index
alternate between $d$ and $d^{p-1}$.  Write
$X_j\in C_e(\mu-1-j,b_j)$ for $j\geqslant0$, with $b_0=t-1$.
The subsequent subset degrees are
\[
 b_{2s-1}=b_{2s}=t-s(p-2)-1
 \qquad\text{for }s\geqslant1,
\]
exactly as long as the corresponding term exists.  For the first kernel calculation, fixing the remaining
$(\mu-2)$-tuple leaves at least $N-(\mu-2)\geqslant2t+1$ subset indices.
At later stages the tuple is shorter, while the sum of the domain and
codomain subset degrees of the tested $d$ or $d^{p-1}$ is at most $2t+1$.
Thus every such application satisfies its cardinality hypothesis.  If the next subset degree would be negative, the preceding
map $d$ is already injective and the descent stops.  For $j\geqslant0$, the sets indexing the summand containing $X_j$ have
size $\mu-1-j+b_j\leqslant t+\mu-2$.  The arc vectors indexed by each such
set are linearly independent by the first bound on $D$ in
\eqref{ex:bounds}, so Lemma~\ref{ex:exact} is available at every
$\partial$-lifting step.

\smallskip\noindent\emph{Tuple degree zero.}
If the descent stops with positive tuple degree, the last $X_j$ is
$\partial$-closed and hence is $\partial$-exact by
Lemma~\ref{ex:exact}.  Otherwise the descent reaches tuple degree $0$.  Let $a$ be the subset
degree at that stage, and let $h\in\{d,d^{p-1}\}$ be the operator for
which $hX$ occurs in the preceding relation.  Because $hX$ is a
$\partial$-boundary, its value on each target set $J$ lies in
$\mathcal I_J$.  Lemma~\ref{sp:repair} changes $X$ by an element of
$\ker h$ so that the corrected array lies in the corresponding
$\mathcal S_a$-space.  Tuple degree zero is reached exactly when the threshold in
\eqref{ex:bottom-bound} is met.  Then Lemma~\ref{sp:repair} requires $D\geqslant2a+e+1$ for $h=d$ and
$D\geqslant2a+p+e-1$ for $h=d^{p-1}$.  At tuple degree zero,
$\mu=2s+1$ gives $a=t-s(p-2)-1$ and $h=d$, hence
$2a+e+1=2t-(\mu-1)(p-2)+e-1$; while $\mu=2s+2$ gives
$a=t-(s+1)(p-2)-1$ and $h=d^{p-1}$, hence
$2a+p+e-1=2t-(\mu-1)(p-2)+e-1$.  Thus in both parity cases the
dimension requirement is exactly~\eqref{ex:bottom-bound}.  If the threshold in~\eqref{ex:bottom-bound} is not met, the descent
stops earlier at positive tuple degree and no tuple-degree-zero adjustment
is needed.  By
Lemma~\ref{ex:exact}, the corrected tuple-degree-zero array is a
$\partial$-boundary.  If $a=0$, the quotient condition at the empty set forces it to be zero,
and we take the zero preimage in negative subset degree.
Similarly, a closed array of positive tuple degree and subset degree
zero is zero by exactness, since its preceding term is zero.

\smallskip\noindent\emph{Returning to $G$.}
Suppose $j\geqslant1$, $\partial X_{j-1}=h_jX_j$ with
$h_j\in\{d,d^{p-1}\}$, and the corrected term is
$\widehat X_j=\partial Y_j$ with $h_j\widehat X_j=h_jX_j$.  Replace
\[
 X_{j-1}\quad\text{by}\quad
 \widehat X_{j-1}=X_{j-1}-h_jY_j,
\]
where the same operator $h_j$ is taken in the shifted subset degree.
Then $\partial\widehat X_{j-1}=0$.  The image in the preceding relation is
unchanged because consecutive maps on the subset index are $d$ and
$d^{p-1}$ in some order and hence their composite is $d^p=0$.
Lemma~\ref{ex:exact} lifts the correction one step upward.  Repeating this
procedure gives $\widehat X_0=\partial C$ while preserving
$d\widehat X_0=dR$.  Hence
\[
 \partial(G-dC)=0.
\]
It is supported on sets of size $t+\mu-1\leqslant D$, so
Lemma~\ref{ex:exact} yields $G-dC=\partial Z$.  If $\mu=e$, the target
space for $Z$ is zero.
\end{proof}

\subsection{Descent in subset degree}

Fix an ordered $\mu$-tuple $\mathbf x$ and regard
$I\mapsto F_{\mathbf x}(I)$ as an $A_{\mathbf x}^{(e)}$-valued function
on subsets; membership in $\mathcal L_b$ is then imposed coefficientwise.
Proposition~\ref{sc:structure}~(\ref{item:sc-low}) writes such a
function as an element of $\operatorname{im}E_{b,b-1}$ plus terms from
degrees $b-mp$.  Lemma~\ref{ex:gapone} removes the
$\operatorname{im}E_{b,b-1}$ term, and the next lemma lowers the largest
remaining subset degree until none remains.

\begin{lemma}\label{loc:auxiliary-claim}
Let $1\leqslant\mu\leqslant e$ and $t\geqslant1$.  Assume the conditions~\eqref{ex:bounds} and~\eqref{ex:bottom-bound}.  If $G\in C_e(\mu,t-1)$ satisfies, for every ordered
$(\mu-1)$-tuple $\mathbf x$,
\[
 (\partial G)_{\mathbf x}\in
 \mathcal L_t(X\setminus\underline{\mathbf x};A_{\mathbf x}^{(e)}),
\]
then there exist $Z\in C_e(\mu+1,t-2)$ and
$C_i\in C_e(\mu,i)$ for $0\leqslant i<t-1$ such that
\[
 G=\partial Z+
 \sum_{0\leqslant i<t-1}E_{t-1,i}C_i.
\]
When $\mu=e$, the term $Z$ is zero; when $t=1$, one has $G=0$.
\end{lemma}

\begin{proof}
For $t=1$, Proposition~\ref{sc:structure}~(\ref{item:sc-low}) gives
$\mathcal L_1=\operatorname{im}E_{1,0}=\operatorname{im}d$.  For one
orientation of each underlying $(\mu-1)$-tuple, choose the constant
coefficient whose $d$-image is $(\partial G)_{\mathbf x}$, and extend
the choices to the other orientations by alternation.  This gives
$R\in C_e(\mu-1,0)$ with $\partial G=dR$.  Lemma~\ref{ex:gapone},
which is valid for $t=1$ under the same bounds, now gives $G=0$.

Assume $t\geqslant2$.  We prove the assertion by induction on $\mu$,
simultaneously for all $t$ at each tuple degree.

\smallskip\noindent\emph{Base of the tuple induction ($\mu=1$).}
Put $y=\partial G$.  Then $y\in\mathcal L_t(X;\Lambda^eV)$ and
$y(J)\in\mathcal I_J$.  For $\mu=1$, the threshold in~\eqref{ex:bottom-bound} is automatic,
and its bound gives $D\geqslant2t+e-1$; hence
Lemma~\ref{sp:supported-low} gives
$y=\sum_{a<t}E_{t,a}c_a$ with $c_a\in\mathcal S_a(X)$ and $c_0=0$.
For $1\leqslant a<t$, Lemma~\ref{ex:exact} gives
$c_a=\partial Y_a$ with $Y_a\in C_e(1,a-1)$.  The arc vectors indexed by
these sets are linearly independent because $a<t\leqslant D$.
Hence
\[
 \partial\Bigl(G-
 \sum_{1\leqslant a<t}E_{t-1,a-1}Y_a\Bigr)=0,
\]
and Lemma~\ref{ex:exact}, now on sets of size $t$, supplies
$Z\in C_e(2,t-2)$.

\smallskip\noindent\emph{Inductive step ($\mu\geqslant2$).}
For one orientation of each underlying $(\mu-1)$-tuple, apply
Proposition~\ref{sc:structure}~(\ref{item:sc-low}) in the coefficient space
$A_{\mathbf x}^{(e)}$,
and extend the chosen coefficient arrays to all orientations by the
alternating sign rule.  Because the inclusion operators act only on the
subset index, this gives global alternating arrays and a decomposition.
The ground set for each tuple has at least $2t$ indices by~\eqref{ex:bounds};
we therefore choose $R\in C_e(\mu-1,t-1)$ and
$B_m\in C_e(\mu-1,t-mp)$, $1\leqslant m\leqslant\lfloor t/p\rfloor$,
so that
\begin{equation}\label{loc:resonance}
 \partial G=dR+
 \sum_{m=1}^{\lfloor t/p\rfloor}E_{t,t-mp}B_m.
\end{equation}
If all additional summands vanish, Lemma~\ref{ex:gapone} gives the
claim.  Otherwise, discard any $B_m$ whose displayed contribution already
vanishes, and set
\[
 s_0=t-p+1,
 \qquad s_m=t-mp+1.
\]
\smallskip\noindent\textbf{Claim (detecting the largest remaining subset degree).}
If $s=s_m$ is largest among the indices for which the
corresponding $B_m$ is nonzero, then for every ordered
$(\mu-2)$-tuple $\mathbf f$,
\begin{equation}\label{loc:highest-detection}
 (\partial B_m)_{\mathbf f}\in
 \mathcal L_s(X\setminus\underline{\mathbf f};A_{\mathbf f}^{(e)}).
\end{equation}

\smallskip\noindent\emph{Proof of the claim.}
Choose
$Y\subseteq X\setminus\underline{\mathbf f}$ of size $2t+1$.  Let
$Q_{a,b}$ be the matrix with rows indexed by $a$-subsets of $Y$, columns
by $b$-subsets, and entry $1$ for disjoint pairs and $0$ otherwise.  Complementing rows identifies the matrix of
$d:F_t(Y;\mathbb F)\to F_{t+1}(Y;\mathbb F)$ with the symmetric matrix $Q_{t,t}$ and the
matrix of $E_{t+1,s_m}$ with $Q_{t,s_m}$.  Let $P$ be the matrix of
$E_{t,s_0}$.  Since
$dE_{t,s_0}=pE_{t+1,s_0}=0$, we have
$Q_{t,t}P=0$ and hence $P^{\mathsf T}Q_{t,t}=0$.  Apply $\partial$ to
\eqref{loc:resonance}, restrict to $Y$, complement the rows, and multiply
by $P^{\mathsf T}$.  If $A$ is an $s_0$-set and $B$ an $s_m$-set, the
$(A,B)$-entry of $P^{\mathsf T}Q_{t,s_m}$ is zero when
$A\cap B\neq\varnothing$, and otherwise equals
\[
 \binom{2t+1-s_0-s_m}{t-s_0}
 =\binom{(m+1)p-1}{p-1}\equiv1\pmod p.
\]
The congruence uses only the lowest base-$p$ digit, which is $p-1$;
the higher digits of $m$ are unrestricted.  It follows that
\[
 \sum_m Q_{s_0,s_m}(\partial B_m)_{\mathbf f}|_Y=0.
\]
Because $s=s_m$ is the largest subset size among the remaining nonzero
terms, every other nonzero summand has column size strictly smaller than
$s$.  Take an $s$-matching $\mathcal P$ and enlarge its support to such a set $Y$.
Since $s_0+s\leqslant2t+1$, there is an $(s_0-s)$-set $R_0$ disjoint
from the matching.
Take the alternating sum over the rows indexed by $R_0\cup J$, where
$J\in\mathcal T(\mathcal P)$.  A column missing a pair cancels between the two
choices in that pair.  A column meeting $R_0$, or containing both members of a pair, is disjoint
from no such row.  The alternating row sum therefore annihilates
every column of size $<s$;
among columns of size $s$, only the sets in $\mathcal T(\mathcal P)$ survive,
with coefficient $(-1)^s$ times their matching sign.  Consequently,
$\Delta_s(\partial B_m)_{\mathbf f}=0$.  Every matching disjoint from
$\mathbf f$ can be included in such a $Y$, so this proves
\eqref{loc:highest-detection} on $X\setminus\underline{\mathbf f}$.
\smallskip\noindent\emph{Lowering the largest subset degree.}
The induction hypothesis now applies to $B_m$ with parameters
$(\mu-1,s,e,D,N)$.  Since $s=t-mp+1$, its first dimension and cardinality
bounds are weaker than~\eqref{ex:bounds} by $mp$ and $2mp-1$.
If the additional condition corresponding to~\eqref{ex:bottom-bound} is
required, then the original parameters satisfy its threshold because
$t-[1+\lceil(\mu-1)/2\rceil(p-2)]\geqslant mp-1-(p-2)\geqslant0$.
The required lower bound on $D$ is smaller than that in
\eqref{ex:bottom-bound} by $(2m-1)p$.  The induction hypothesis therefore
applies.
If $s=1$ it gives
$B_m=0$.  If $s\geqslant2$, choose $Z_m\in C_e(\mu,s-2)$ and
$C_i\in C_e(\mu-1,i)$ for $0\leqslant i<s-1$ so that
\[
 B_m=\partial Z_m+\sum_{0\leqslant i<s-1}E_{s-1,i}C_i.
\]
Subtract $E_{t-1,s-2}Z_m$ from $G$.  The boundary term in that
summand disappears by~\eqref{ex:commutation}.  Every remaining new term
has the form
\[
 \lambda_iE_{t,i}C_i,
 \qquad
 \lambda_i=\binom{t-i}{s-1-i},\quad i<s-1.
\]
If $p\nmid t-i$, absorb it into $dR$ using
$dE_{t-1,i}=(t-i)E_{t,i}$; if $p\mid t-i$, retain it as a new
term of strictly smaller subset degree.  Combine terms of the same degree and discard zero terms.  The largest
remaining nonzero coefficient degree then strictly decreases; no uniqueness
of the representation is needed.  Since this degree is a nonnegative
integer, finitely many steps reduce
\eqref{loc:resonance} to
$\partial G'=dR'$.  Lemma~\ref{ex:gapone} applies to $G'$, and restoring
the terms coming from smaller subset degrees proves the lemma.
\end{proof}

\begin{proof}[Proof of Theorem~\ref{loc:slr}]
Define $G\in C_e(e,t-1)$ from the functions $h_U$.  We verify the
hypothesis of Lemma~\ref{loc:auxiliary-claim}.  For an ordered $e$-tuple
$\mathbf u=(u_1,\ldots,u_e)$ and
$I\in\binom{X\setminus\underline{\mathbf u}}{t-1}$, set
\[
 G_{\mathbf u}(I)
 =h_{\underline{\mathbf u}}(I)w_{\mathbf u}.
\]
The function depends only on the underlying $e$-set while
$w_{\mathbf u}$ is alternating, so $G\in C_e(e,t-1)$.  For a fixed
ordered $(e-1)$-tuple $\mathbf a$, the linear map
\[
 V\longrightarrow\Lambda^eV,
 \qquad v\longmapsto w_{\mathbf a}\wedge v,
\]
has kernel $W_{\underline{\mathbf a}}$ because the $e-1$ vectors indexed
by $\mathbf a$ are independent.  Hence~\eqref{loc:geometric-equation}
is exactly $\Delta_t(\partial G)_{\mathbf a}=0$.  Since
$N-(e-1)\geqslant2t$, Proposition~\ref{sc:structure}~(\ref{item:sc-low}) gives
\[
 (\partial G)_{\mathbf a}\in
 \mathcal L_t(X\setminus\underline{\mathbf a};A_{\mathbf a}^{(e)}).
\]
Apply Lemma~\ref{loc:auxiliary-claim} with $\mu=e$.  Its cardinality
bound is the first assumption of the theorem.  Together, its two dimension requirements are exactly
$D\geqslant t+e-1+\max\{0,t-(e-1)(p-2)\}$.
Since $C_e(e+1,\cdot)=0$, the boundary term
vanishes and $G$ is a sum of images from smaller subset degrees.  For a fixed
$e$-set $U$, choose its inherited order $\mathbf u$.  The one-dimensional
space $A_{\mathbf u}^{(e)}=\mathbb Fw_{\mathbf u}$ identifies the resulting
statement with~\eqref{loc:slr-conclusion}.
\end{proof}

Section~\ref{gl:section} uses only Theorem~\ref{loc:slr}.  The proof just
completed combines Lemma~\ref{ex:exact} with
Proposition~\ref{sc:structure}~(\ref{item:sc-kernel-d})--(\ref{item:sc-kernel-dp})
to rewrite the matching equations using coefficients indexed by smaller
subsets of each complement $X\setminus U$.

\section{Global vanishing and determinant relations}\label{gl:section}

Throughout this section, let $\mathbb F$ be a field of characteristic $p>0$.  We first
use Theorem~\ref{loc:slr} to prove a vanishing statement for coefficient
arrays.  A linear system associated with an arc then converts that vanishing
statement into determinant relations and, ultimately, our full-row-rank
theorem for Chowdhury's matrices.

\subsection{Vanishing of coefficient arrays}

The next lemma combines vanishing on tuples meeting $J$, a vector-valued
linear relation, and membership in $\mathcal L_T$.  Its proof successively
lowers the degree of a representation by inclusion maps until all
coordinates with distinct indices vanish.  The first two conditions then
imply vanishing for repeated indices as well.

\begin{lemma}\label{gl:array}
Let $T\geqslant1$ and $1\leqslant e\leqslant D$, and let $(w_x)_{x\in X}$ be an arc in a
$D$-dimensional $\mathbb F$-vector space $W$.  Put $N:=\lvert X\rvert$, and suppose
\[
 N\geqslant2T+e,\qquad
 D\geqslant T+e-1+\max\{0,T-(e-1)(p-2)\}.
\]
For each $J\in\binom XT$, let $E^J:X^e\to\mathbb F$ be a symmetric function of its $e$ arguments; the arguments need not be distinct.  Write
$E^J_{i_1,\ldots,i_e}:=E^J(i_1,\ldots,i_e)$, and, for $U\in\binom Xe$, let $E^J_U$ denote the common value of $E^J$ on any ordering of the elements of $U$.  Assume:
\begin{enumerate}[label=\textup{(\roman*)},ref=\roman*]
\item\label{item:gl-array-support}
$E^J_{i_1,\ldots,i_e}=0$ whenever one of $i_1,\ldots,i_e$ belongs to
$J$;
\item\label{item:gl-array-balance}
For every $J\in\binom XT$ and every $A\in X^{e-1}$,
$\sum_{z\in X}E^J_{A,z}w_z=0$.  Here $X^0$ consists of the empty tuple; thus, when $e=1$, condition~\textup{(\ref{item:gl-array-balance})} reads $\sum_{z\in X}E^J_z w_z=0$.
\item\label{item:gl-array-low}
For every $e$-subset $U\subseteq X$, the function
$J\mapsto E^J_U$ on $\binom{X\setminus U}T$ belongs to
$\mathcal L_T(X\setminus U;\mathbb F)$.
\end{enumerate}
Then $E^J$ vanishes identically for every $J\in\binom XT$.
\end{lemma}

\begin{proof}
For $U\in\binom Xe$, put $f_U(J)=E^J_U$.  For $-1\leqslant a<T$,
define
\[
 \mathcal F_a(U)=\sum_{0\leqslant i\leqslant a}\operatorname{im}E_{T,i}
 \subseteq F_T(X\setminus U;\mathbb F),
\]
with $\mathcal F_{-1}(U)=0$.  Since $\lvert X\setminus U\rvert\geqslant2T$,
Proposition~\ref{sc:structure}~(\ref{item:sc-low}) gives
$\mathcal F_{T-1}(U)=\mathcal L_T(X\setminus U;\mathbb F)$, so
condition~\textup{(\ref{item:gl-array-low})} starts the induction at
$a=T-1$.  It remains to prove the implication
\[
 \bigl[f_U\in\mathcal F_a(U)\text{ for every }U\in\tbinom Xe\bigr]
 \quad\Longrightarrow\quad
 \bigl[f_U\in\mathcal F_{a-1}(U)\text{ for every }U\in\tbinom Xe\bigr]
 \qquad\text{for }0\leqslant a<T.
\]
Iteration from $a=T-1$ to $a=-1$ then makes all distinct-index
coordinates vanish.  The repeated-index coordinates are recovered from
conditions~\textup{(\ref{item:gl-array-support})} and
\textup{(\ref{item:gl-array-balance})}.  We first verify this endpoint and then prove the displayed implication.

\smallskip\noindent\emph{From distinct-index vanishing to total vanishing.}
Suppose $f_U=0$ for every $U$, and fix $J\in\binom XT$.
The support of a tuple is the set of its distinct entries.
Every coordinate with $e$ distinct indices is zero: condition
\textup{(\ref{item:gl-array-support})} applies if its support meets $J$,
and otherwise $f_U(J)=0$.  Now descend on the support size.  Suppose all coordinates of support size greater
than $s$ vanish, and take a tuple with support $S$, $\lvert S\rvert=s<e$.
Coordinates with $S\cap J\ne\varnothing$ vanish by condition~\textup{(\ref{item:gl-array-support})}.
Otherwise delete one occurrence of a repeated entry to obtain an
$(e-1)$-tuple $A$ with the same support.  Condition~\textup{(\ref{item:gl-array-balance})} gives
\[
 \sum_{z\in S}E^J_{A,z}w_z=0.
\]
Terms with $z\in J$ vanish by condition~\textup{(\ref{item:gl-array-support})}, and those with
$z\notin S\cup J$ have larger support.  Since $\lvert S\rvert\leqslant e-1\leqslant D$,
the vectors indexed by $S$ are independent, so every displayed
coefficient vanishes.  This proves $E^J=0$.  For $e=1$, the initial
distinct-index argument already proves the assertion.

\smallskip\noindent\emph{The step $\mathcal F_a\to\mathcal F_{a-1}$.}
Let $0\leqslant a<T$ and suppose $f_U\in\mathcal F_a(U)$ for every $U$.
For each $U$, choose once and for all functions
$b_U^i:\binom{X\setminus U}{i}\to\mathbb F$, $0\leqslant i\leqslant a$,
such that
\[
 f_U=\sum_{i=0}^{a}E_{T,i}b_U^i.
\]
These choices remain fixed when we restrict to smaller sets.
Put $W_B=\operatorname{span}\{w_x:x\in B\}$ for $B\subseteq X$, and set
\[
 D_a=D-T+a+1,\qquad
 \lvert R\rvert=T-a-1,\qquad \lvert S\rvert=\max\{2a+e+1,D_a\}.
\]
Disjoint sets $R,S\subseteq X$ of these sizes exist because
\[
 \lvert R\rvert+\lvert S\rvert=\max\{T+a+e,D\}\leqslant N.
\]
Here $T+a+e\leqslant2T+e-1<N$, and $D\leqslant N$ by the arc
convention.  In $W/W_R$, the vectors indexed by $S$ form an arc of dimension $D_a$:
any $D_a$ of them, together with the vectors indexed by $R$, give $D$
linearly independent vectors of the original arc.

Fix an ordered $(e-1)$-tuple $\mathbf a$ of distinct elements of $S$
and an $(a+1)$-matching $\mathcal P$ in
$S\setminus\underline{\mathbf a}$.  For
$Y\in\mathcal T(\mathcal P)$ put $J=R\cup Y$.
Take the alternating sum of
condition~\textup{(\ref{item:gl-array-balance})} over
$Y\in\mathcal T(\mathcal P)$, modulo
$W_{R\cup\underline{\mathbf a}}$.  For coefficient degree $i$,
the contribution before taking this sum is
\[
 \sum_{z\notin J\cup\underline{\mathbf a}}\sum_{I\in\binom Ji}
 b^i_{\underline{\mathbf a}\cup\{z\}}(I)w_z
 =\sum_{I\in\binom Ji}\sum_{z\notin I\cup\underline{\mathbf a}}
 b^i_{\underline{\mathbf a}\cup\{z\}}(I)w_z
 -\sum_{I\in\binom Ji}\sum_{z\in J\setminus I}
 b^i_{\underline{\mathbf a}\cup\{z\}}(I)w_z.
\]
Each summand in the first sum depends on at most $i$ matching choices,
and each summand in the second on at most $i+1$.  Thus every degree $i<a$ cancels, as does
the first sum at $i=a$.  In the remaining sum at $i=a$, a subset index
meeting $R$ still misses a matching choice, and terms with $z\in R$
vanish in the quotient.  We obtain
\[
 \sum_{Y\in\mathcal T(\mathcal P)}\epsilon_{\mathcal P}(Y)
 \sum_{z\in Y}b^a_{\underline{\mathbf a}\cup\{z\}}
       (Y\setminus\{z\})w_z
 \in W_{R\cup\underline{\mathbf a}}.
\]
These are precisely the equations of Theorem~\ref{loc:slr} for the
restricted functions $b_U^a$ on the projected arc indexed by $S$, with
$t'=a+1$.  Its cardinality hypothesis holds because
$\lvert S\rvert\geqslant2t'+e-1$, and
\[
 D_a\geqslant a+e,
 \qquad
 D_a\geqslant2a+e+1-(e-1)(p-2).
\]
The first inequality follows from $D\geqslant T+e-1$; the second follows
from $D\geqslant2T+e-1-(e-1)(p-2)$ and $a\leqslant T-1$.  Thus the dimension hypothesis of Theorem~\ref{loc:slr} also holds, and
the theorem gives
\[
 b_U^a|_{\binom{S\setminus U}{a}}
 \in\mathcal L_a(S\setminus U;\mathbb F)
 \qquad\text{for }U\in\binom Se.
\]

For $a>0$, any $a$-matching in $X\setminus U$, together with $U$, can
be included in such a set $S$; the same size calculation then provides a
disjoint $R$.  Thus
$b_U^a\in\mathcal L_a(X\setminus U;\mathbb F)$ on the full complement
$X\setminus U$.
Proposition~\ref{sc:structure}~(\ref{item:sc-low}) writes
$b_U^a=\sum_{i<a}E_{a,i}c_U^i$, and
\[
 E_{T,a}E_{a,i}=\binom{T-i}{a-i}E_{T,i}
\]
gives $f_U\in\mathcal F_{a-1}(U)$ without division.
For $a=0$, choose $S$ containing $U$; the $t'=1$ conclusion of
Theorem~\ref{loc:slr} gives $b_U^0(\varnothing)=0$, so
$f_U\in\mathcal F_{-1}(U)$.
Starting from condition~\textup{(\ref{item:gl-array-low})}, iteration gives
$f_U\in\mathcal F_{-1}(U)=0$ for every $U$.  The argument for repeated
indices then gives $E^J=0$ for every $J$, proving the lemma.
\end{proof}

\subsection{A linear system associated with an arc}

We next introduce a scalar system associated with an arc and prove that
it has only the zero solution under a cardinality hypothesis.  In the
next subsection, projection turns a hypothetical nonzero determinant
relation into a nonzero solution of this system.

Let $r,e\geqslant1$ and let $H=(v_x)_{x\in X}$ be an arc in an
$(r+1)$-dimensional $\mathbb F$-vector space $V$.
For $C\in\binom Xr$, choose $0\ne\alpha_C\in V^*$ with
$\ker\alpha_C=\operatorname{span}\{v_c:c\in C\}$.
Let $D_r^{(e)}(H)$ be the homogeneous system in the variables $x_C$,
$C\in\binom Xr$.  For every $r$-matching $\mathcal M$ and every $e$-subset $U$ disjoint
from its support, impose the equation
\begin{equation}\label{gl:normal-equations}
 \sum_{C\in\mathcal T(\mathcal M)}
 \epsilon_{\mathcal M}(C)x_C\prod_{u\in U}\alpha_C(v_u)=0.
\end{equation}
Write $\ker D_r^{(e)}(H)$ for the solution space.  Replacing $\alpha_C$ by $s_C\alpha_C$, with $s_C\ne0$, corresponds to
replacing $x_C$ by $s_C^{-e}x_C$.  Hence the triviality of the kernel is
independent of these choices.

\begin{theorem}\label{gl:zero-kernel}
Let $\mathbb F$ be a field of characteristic $p>0$, let $e,r\geqslant1$,
and let $H=(v_x)_{x\in X}$ be an arc in an $(r+1)$-dimensional
$\mathbb F$-vector space.  If
\[
 \lvert X\rvert\geqslant2r+e+\max\{0,r-(e-1)(p-2)\},
\]
then $\ker D_r^{(e)}(H)=\{0\}$.
\end{theorem}

\begin{proof}
Put $N=\lvert X\rvert$, and let $x\in\ker D_r^{(e)}(H)$.  We apply
Lemma~\ref{gl:array} with $T=r$ in four steps.

\noindent\emph{Step 1: multilinear forms.}
Define symmetric $e$-linear forms
\[
 \Phi_J(\xi_1,\ldots,\xi_e)
 =x_J\prod_{j=1}^e\alpha_J(\xi_j),
 \qquad J\in\binom Xr.
\]
No division by $e!$ is involved.

\noindent\emph{Step 2: an arc in $K^*$.}
For the space $\mathbb F^X$ of coordinate families $c=(c_y)_{y\in X}$,
define
\[
 \nu:\mathbb F^X\longrightarrow V,
 \qquad \nu(c)=\sum_{y\in X}c_yv_y,
 \qquad K=\ker\nu.
\]
Because the arc spans $V$, the sequence
\[
\begin{tikzcd}[column sep=small]
 0 \arrow[r] & K \arrow[r,hook] & \mathbb F^X
 \arrow[r,"\nu"] & V \arrow[r] & 0
\end{tikzcd}
\]
is exact.  Thus $\dim K^*=N-r-1$.  For $y\in X$, define
$w_y\in K^*$ by $w_y(c)=c_y$.
The functionals $(w_y)_{y\in X}$ form an arc in $K^*$.  To see this,
fix $B\in\binom X{N-r-1}$.  The coordinate restriction
$K\to\mathbb F^B$ is injective: a vector in its kernel is supported on the
$r+1$ indices in $X\setminus B$, whose arc vectors are linearly
independent.  The dimensions agree, so the restriction is an isomorphism;
hence $(w_y)_{y\in B}$ is a basis of $K^*$.
Also, for every $\ell\in V^*$,
\[
 \sum_{y\in X}\ell(v_y)w_y=0,
\]
because evaluation on $c\in K$ gives $\ell(\sum_y c_yv_y)=0$.

\noindent\emph{Step 3: the three hypotheses of Lemma~\ref{gl:array}.}
For $J\in\binom Xr$, define
\[
 E^J(i_1,\ldots,i_e)=\Phi_J(v_{i_1},\ldots,v_{i_e}),
 \qquad\text{for }(i_1,\ldots,i_e)\in X^e.
\]
Condition~\textup{(\ref{item:gl-array-support})} holds because
$\alpha_J(v_j)=0$ for $j\in J$.  For
condition~\textup{(\ref{item:gl-array-balance})}, apply the preceding
identity to
$\ell(v)=\Phi_J(v_{i_1},\ldots,v_{i_{e-1}},v)$; when $e=1$, take
$\ell=\Phi_J$.
For condition~\textup{(\ref{item:gl-array-low})}, fix $U=\{u_1,\ldots,u_e\}\subseteq X$ in inherited
order.  For any $r$-matching $\mathcal M$ disjoint from $U$,
\[
 (\Delta_r(J\mapsto E^J_U))(\mathcal M)
 =\sum_{J\in\mathcal T(\mathcal M)}\epsilon_{\mathcal M}(J)x_J
      \prod_{u\in U}\alpha_J(v_u)=0
\]
by~\eqref{gl:normal-equations}.  Hence $J\mapsto E^J_U$ belongs to
$\mathcal L_r(X\setminus U;\mathbb F)$.

\noindent\emph{Step 4: conclusion.}
The hypothesis of Theorem~\ref{gl:zero-kernel} gives $N\geqslant2r+e$ and
\[
 \dim K^*=N-r-1
 \geqslant r+e-1+\max\{0,r-(e-1)(p-2)\}.
\]
Lemma~\ref{gl:array} gives $E^J=0$ for every $J$.  Since the $v_y$
span $V$, multilinearity then gives $\Phi_J=0$.  Choose $v$ with
$\alpha_J(v)\ne0$; from
$0=\Phi_J(v,\ldots,v)=x_J\alpha_J(v)^e$ we obtain $x_J=0$.
\end{proof}

\subsection{Projection of determinant relations}

We now convert determinant relations into the scalar system of the
preceding subsection.  Complementing subsets first expresses the
coefficients using smaller subsets.  Projection from a suitable set of
arc points then gives a solution of $D_\rho^{(e)}(H_D)$, which
Theorem~\ref{gl:zero-kernel} forces to vanish.

For an ordered $k$-dimensional arc $(v_x)_{x\in X}$, an index $u\in X$,
and $B\in\binom X{k-1}$, write $\det(v_u,v_B)$ for the determinant whose
first row is $v_u$ and whose remaining rows are indexed by $B$ in
inherited order.

\begin{theorem}\label{gl:extension-obstruction}
Let $\mathbb F$ be a field of characteristic $p>0$, let $e,r\geqslant1$,
and set $k:=p+r$.  Let $G=(v_x)_{x\in X}$ be a $k$-dimensional arc with
\[
 k+e-1\leqslant \lvert X\rvert
 \leqslant k+p+e-3+\min\{r,e(p-2)\}.
\]
If scalars $\lambda_B\in\mathbb F$, indexed by $B\in\binom X{k-1}$,
satisfy, for every $e$-subset $U\subseteq X$ and every
$(k-2)$-subset $S\subseteq X\setminus U$,
\begin{equation}\label{gl:extension-star}
 \sum_{\substack{B\in\binom{X\setminus U}{k-1}\\S\subseteq B}}
 \lambda_B\prod_{u\in U}\det(v_u,v_B)=0,
\end{equation}
then $\lambda_B=0$ for every $B\in\binom X{k-1}$.
\end{theorem}

\begin{proof}
Suppose that $\lambda$ satisfies~\eqref{gl:extension-star} and that
$\lambda_{B_0}\ne0$ for some $B_0\in\binom X{k-1}$.  Put
\[
 \rho=\lvert X\rvert-(k+p+e-3),\qquad h=p+\rho-2.
\]
The size assumptions give
$2-p\leqslant\rho\leqslant\min\{r,e(p-2)\}$.
For $U\in\binom Xe$, set $Y=X\setminus U$ and define
\[
 a_U(B)=\lambda_B\prod_{u\in U}\det(v_u,v_B),
 \qquad B\in\binom Y{k-1}.
\]
The complement of a $(k-1)$-subset of $Y$ has size $h$.
Thus the function
\[
 a_U^\vee:\binom Yh\longrightarrow\mathbb F,
 \qquad a_U^\vee(Q)=a_U(Y\setminus Q)
\]
satisfies $d a_U^\vee=0$ by~\eqref{gl:extension-star}:
complementing a $(k-2)$-subset converts its supersets into the
$h$-subsets of an $(h+1)$-set.  Since
\[
 \lvert Y\rvert-(2h+1)=r-\rho\geqslant0,
\]
Proposition~\ref{sc:structure}~(\ref{item:sc-kernel-d}) applies.  Since
$h=p+\rho-2$, the map is
injective for $\rho\leqslant0$, whereas for $\rho\geqslant1$ its kernel
comes from degree $\rho-1$.

\smallskip\noindent\emph{Case 1: injectivity ($\rho\leqslant0$).}
Then $0\leqslant h\leqslant p-2$, so part~(\ref{item:sc-kernel-d}) gives
$a_U^\vee=0$ for every $U$.  Choose
$U\in\binom{X\setminus B_0}e$, possible because
$\lvert X\rvert\geqslant k+e-1$.  The arc property gives
\[
 0=a_U(B_0)=\lambda_{B_0}\prod_{u\in U}\det(v_u,v_{B_0})\ne0,
\]
a contradiction.

\smallskip\noindent\emph{Case 2: projection to Theorem~\ref{gl:zero-kernel} ($\rho\geqslant1$).}
Because $h+1-p=\rho-1$, there exists
$g_U\in F_{\rho-1}(Y;\mathbb F)$ such that
\[
 a_U^\vee=E_{h,\rho-1}g_U.
\]
Fix any $D\in\binom X{k-1-\rho}$ and project from
$W_D=\operatorname{span}\{v_d:d\in D\}$.  In
$\overline V_D=\mathbb F^k/W_D$, the family
$H_D=(\overline v_x)_{x\in X\setminus D}$ is an arc of dimension
$\rho+1$ and size $p+2\rho+e-2$.
For $C\in\binom{X\setminus D}\rho$, define the sign
\[
 \varepsilon_D(C)=(-1)^{\lvert \{(c,d)\in C\times D:c<d\}\rvert},
\]
which changes the inherited order on $D\cup C$ to the concatenation
of the inherited orders on $D$ and $C$.  Set
\[
 \alpha_C(\overline v)=\varepsilon_D(C)\det(v,v_{D\cup C}),
 \qquad x_C=\varepsilon_D(C)^e\lambda_{D\cup C}.
\]
Adding an element of $W_D$ to $v$ leaves the determinant unchanged,
so $\alpha_C$ is well defined.  It is nonzero by the arc property,
and its kernel is spanned by the $\rho$ independent projected vectors
indexed by $C$.  These functionals and $x$ are fixed independently of $U$.

For $U\in\binom{X\setminus D}e$, put $Z=Y\setminus D$.
Then $\lvert Z\rvert=p+2\rho-2\geqslant2\rho$.  For $C\in\binom Z\rho$,
\[
 a_U(D\cup C)
 =\sum_{\substack{I\in\binom Z{\rho-1}\\I\cap C=\varnothing}}g_U(I)
 =\sum_{I\in\binom Z{\rho-1}}g_U(I)\sum_{J\subseteq I\cap C}(-1)^{\lvert J\rvert}
 =\sum_{j=0}^{\rho-1}\sum_{J\in\binom Cj}(-1)^j
  \sum_{\substack{I\in\binom Z{\rho-1}\\J\subseteq I}}g_U(I).
\]
For fixed $j<\rho$, the $j$-th summand is $(E_{\rho,j}h_j)(C)$, where
$h_j(J)=(-1)^j\sum_{I\in\binom Z{\rho-1},\,J\subseteq I}g_U(I)$.
Therefore
\[
 C\longmapsto a_U(D\cup C)
 \in\sum_{j<\rho}\operatorname{im}E_{\rho,j}
 =\mathcal L_\rho(Z;\mathbb F)
\]
by~\eqref{sc:all-low}.  Since
\[
 a_U(D\cup C)=x_C\prod_{u\in U}\alpha_C(\overline v_u),
\]
applying $\Delta_\rho$ to this function gives precisely the equations of
$D_\rho^{(e)}(H_D)$ for this $U$.  Therefore
$x\in\ker D_\rho^{(e)}(H_D)$.
This calculation uses only $1\leqslant\rho\leqslant r$; it applies to
arbitrary coefficients and any such $D$, without the assumption
$\rho\leqslant e(p-2)$.

The remaining inequality $\rho\leqslant e(p-2)$ now gives
\[
 \max\{0,\rho-(e-1)(p-2)\}\leqslant p-2.
\]
The displayed inequality shows that
$\lvert H_D\rvert=p+2\rho+e-2$ satisfies the cardinality hypothesis of
Theorem~\ref{gl:zero-kernel}; hence $x=0$.
Choosing $D\subseteq B_0$ of size $k-1-\rho$ yields
$C_0=B_0\setminus D$ and
$x_{C_0}=\varepsilon_D(C_0)^e\lambda_{B_0}\ne0$, a contradiction.
\end{proof}

\subsection{Chowdhury's matrices and full row rank}

The determinant relations of Theorem~\ref{gl:extension-obstruction} are
exactly the left-kernel equations for Chowdhury's matrices.  We now record
the matrix used in Section~\ref{sec:mds-deduction}.

Let $G=(v_x)_{x\in X}$ be an ordered $k$-dimensional arc over a field
$\mathbb F$.  For $n\geqslant0$, and using the determinant convention of
the preceding subsection, Chowdhury's matrix $M_G^{\uparrow n}$ has rows
indexed by $C\in\binom X{k-1}$ and columns indexed by pairs $(U,A)$ with
$U\in\binom Xn$ and $A\in\binom{X\setminus U}{k-2}$.  Its entries are
\[
 M_G^{\uparrow n}(C;(U,A)):=
 \begin{cases}
 \displaystyle\prod_{u\in U}\det(v_u,v_C),
 &A\subseteq C\subseteq X\setminus U,\\[1ex]
 0,&\text{otherwise}.
 \end{cases}
\]
The empty product is $1$.  A change of coordinates multiplies every
nonzero entry in a fixed row by the same nonzero scalar.  Reordering the arc
only permutes rows and columns and changes signs.  The rank is therefore independent of both choices.  For $n=0$, this is the ordinary inclusion
matrix $E_{k-1,k-2}$ on the index set $X$.

\begin{theorem}[Full-row-rank theorem]\label{gl:full-rank}
Let $\mathbb F$ have characteristic $p>0$, let $n\geqslant0$, and let
$G$ be a $k$-dimensional arc over $\mathbb F$, $k\geqslant2$, of size
$m$.  If
\begin{equation}\label{gl:full-rank-range}
 k+n-1\leqslant m\leqslant
 k+n-1+\min\{k-2,(n+1)(p-2)\},
\end{equation}
then $M_G^{\uparrow n}$ has full row rank over $\mathbb F$.
\end{theorem}

At $m=2k-3+n$, this gives Chowdhury's conjectured range without
its finite-field size restriction.

\begin{corollary}[Chowdhury's full-row-rank conjecture]
\label{gl:chowdhury-full-rank}
Let $\mathbb F$ have characteristic $p>0$, and let $n\geqslant0$ and
$k\geqslant2$ satisfy $k\leqslant p+n(p-2)$.  For every
$k$-dimensional arc $G$ over $\mathbb F$ of size $2k-3+n$, the matrix
$M_G^{\uparrow n}$ has full row rank.
\end{corollary}

For $\mathbb F=\mathbb F_q$, this proves
\cite[Conjecture~1.9]{Chowdhury2016}: its additional condition
$k\leqslant(q+4-n)/2$ is equivalent to $2k-3+n\leqslant q+1$ and is
not needed for the rank statement itself.

\begin{lemma}[Rank inheritance]\label{gl:inheritance}
Fix $k\geqslant2$, $j\geqslant0$, and $s\geqslant k$.  Suppose
$M_G^{\uparrow j}$ has full row rank for every $k$-dimensional arc $G$
over $\mathbb F$ of size $s$.  Then $M_H^{\uparrow(j+1)}$ has full row
rank for every $k$-dimensional arc $H$ over $\mathbb F$ of size $s+1$.
\end{lemma}

\begin{proof}
Write $H=(v_y)_{y\in Y}$ and fix a row $C$.  Choose
$x\in Y\setminus C$ and restrict to columns with
$x\in U$, say $U=\{x\}\cup U'$ with $\lvert U'\rvert=j$.  These columns vanish
on rows containing $x$, while on the rows avoiding $x$ they form
\[
 \operatorname{diag}(\det(v_x,v_{C'}))\,
 M_{H\setminus\{x\}}^{\uparrow j}.
\]
The diagonal factor is invertible because $H$ is an arc, and the second
factor has full row rank by hypothesis.  Their span therefore contains the unit vector indexed by $C$.  Since $C$ was arbitrary, the matrix has full
row rank.
\end{proof}

\begin{proof}[Proof of Theorem~\ref{gl:full-rank}]
We treat three cases, using the $n=0$ inclusion rank, rank inheritance,
and Theorem~\ref{gl:extension-obstruction}, respectively.

First let $n=0$.  Complementing rows and columns identifies the transpose
of $M_G^{\uparrow0}=E_{k-1,k-2}$ with
$E_{m-k+2,m-k+1}$.  Put $a=m-k+1$.  The upper bound in
\eqref{gl:full-rank-range} gives
$a\leqslant\min\{k-2,p-2\}$, hence $a<p-1$, and also
$m\leqslant2k-3$, equivalently $m\geqslant2a+1$.
Proposition~\ref{sc:structure}~(\ref{item:sc-kernel-d}) therefore makes
$E_{a+1,a}$ injective, so $M_G^{\uparrow0}$ has full row rank.

Now assume $n\geqslant1$ and $k\leqslant p$.  In this case
\eqref{gl:full-rank-range} becomes $k+n-1\leqslant m\leqslant2k+n-3$.
If $m=k+n-1$, then for each row $C$ the complement
$U=X\setminus C$ has size $n$, and any column $(U,A)$ with
$A\in\binom C{k-2}$ is supported on that row alone and has nonzero
entry there.  Thus the matrix has full row rank.  If $m\geqslant k+n$,
then $k\leqslant m-n\leqslant2k-3$, so the $n=0$ case applies to every
$(m-n)$-point subarc; iterating Lemma~\ref{gl:inheritance} $n$ times
gives the result.

Finally suppose $n\geqslant1$ and $k>p$.  Put $r=k-p$ and $e=n$.
The upper endpoint in \eqref{gl:full-rank-range} is
\[
 k+p+n-3+\min\{k-p,n(p-2)\},
\]
which is exactly the upper endpoint in
Theorem~\ref{gl:extension-obstruction}.  Any left-kernel vector of $M_G^{\uparrow n}$ satisfies the determinant
relations~\eqref{gl:extension-star}.  Theorem~\ref{gl:extension-obstruction}
therefore makes the left kernel zero, so $M_G^{\uparrow n}$ has full row
rank.
\end{proof}

\begin{proof}[Proof of Corollary~\ref{gl:chowdhury-full-rank}]
At $m=2k-3+n$, condition \eqref{gl:full-rank-range} is equivalent to
$k-2\leqslant(n+1)(p-2)$, that is,
$k\leqslant p+n(p-2)$.  Theorem~\ref{gl:full-rank} applies.
\end{proof}

\section{Finite-field deduction of the MDS bound}
\label{sec:mds-deduction}

The results of Sections~\ref{sc:section}--\ref{gl:section} hold over an
arbitrary field of positive characteristic.  We now specialize to
$\mathbb F_q$.  To prove the first interval in Theorem~\ref{thm:mds-range},
we use Lemma~\ref{mds:ball-lavrauw-kernel}: a hypothetical $(q+2)$-point
arc gives a nonzero left-kernel vector for a matrix that
Theorem~\ref{gl:full-rank} forces to have full row rank.

\begin{lemma}\label{mds:ball-lavrauw-kernel}
Let $q$ be a prime power, let $k\geqslant3$, and let
$A=(v_x)_{x\in X}$ be a $k$-dimensional arc over $\mathbb F_q$ of size
$q+2$.  Put $e:=q-2k+5$ and suppose $e\geqslant0$.  Then
$M_A^{\uparrow e}$ has a nonzero left-kernel vector; indeed, one may
choose a vector all of whose coordinates are nonzero.
\end{lemma}

\begin{proof}
Set $t=k-3$ and apply the construction in the proof of
\cite[Theorem~12.1]{BallLavrauw2019} with $G=A$ and $n=e$.
The required numerical identities are
$\lvert G\rvert=q+2=k+t+e$ and $q+k-1-t=q+2$, so the containing arc can
be taken to be $A$.  The construction gives a solution
$(\lambda_B)_{B\in\binom X{k-1}}$ of the system $P_e\lambda=0$ defined
there, with every coordinate nonzero.  Changing representatives
$v_x\mapsto c_xv_x$ only diagonally rescales $M_A^{\uparrow e}$: row $B$
by $(\prod_{b\in B}c_b)^e$ and column $(U,S)$ by $\prod_{u\in U}c_u$.
The existence of an all-nonzero left-kernel vector is therefore invariant
under the choice of representatives.  After the inverse row rescaling, we
again denote the vector by $\lambda$.  The $P_e$-row indexed by $(S,U)$ is precisely the pairing of $\lambda$
with the column $(U,S)$ of $M_A^{\uparrow e}$; terms with
$B\cap U\ne\varnothing$ vanish because the corresponding determinant
product is zero.  Thus $\lambda$ is a nonzero left-kernel vector.
\end{proof}

\smallskip\noindent The choice $e=q-2k+5$ gives
\begin{equation}\label{mds:parameter-bridge}
 q+2=2k-3+e,
 \qquad
 k\leqslant B(p,q)
 \quad\Longleftrightarrow\quad
 k\leqslant p+e(p-2).
\end{equation}
For $p\geqslant3$ and $k\geqslant3$, the inequality on the right also
forces $e\geqslant0$.  Thus, for $k\geqslant3$, the condition defining the
first interval in Theorem~\ref{thm:mds-range} is precisely the condition in
Corollary~\ref{gl:chowdhury-full-rank} at $\lvert A\rvert=2k-3+e$.

\begin{proof}[Proof of Theorem~\ref{thm:mds-range}]
The standard $q+1$ vectors
$(1,t,t^2,\ldots,t^{k-1})$, $t\in\mathbb F_q$, together with
$(0,\ldots,0,1)$ form a $k$-dimensional arc by the Vandermonde
determinant.  It remains to prove the upper bound.

First suppose $2\leqslant k\leqslant B(p,q)$.  For $k=2$, an arc
contains at most one vector from each one-dimensional subspace of
$\mathbb F_q^2$, hence has size at most $q+1$.  Let $k\geqslant3$ and
put $e=q-2k+5$.  If a $(q+2)$-point arc $A$ existed, then
Lemma~\ref{mds:ball-lavrauw-kernel} would give a nonzero left-kernel
vector of $M_A^{\uparrow e}$.  By \eqref{mds:parameter-bridge}, however,
$\lvert A\rvert=2k-3+e$ and $k\leqslant p+e(p-2)$, so
Theorem~\ref{gl:full-rank} gives full row rank, a contradiction.  This
proves the first interval in~\eqref{intro:main-range}.

Finally suppose $q+2-B(p,q)\leqslant k\leqslant q$ and that a
$(q+2)$-point arc exists in dimension $k$.  Viewing the arc as the columns
of a generator matrix, the classical duality of linear MDS codes
\cite[p.~117]{Singleton1964} gives an arc of the same size in dimension
$q+2-k$.  Here $2\leqslant q+2-k\leqslant B(p,q)$, contradicting the
first interval just proved.  This proves the second interval in~\eqref{intro:main-range}.
\end{proof}

\section{Applications to normal rational curves and projective Reed--Solomon codes}
\label{sec:nrc-prs-applications}

We begin by recalling the conventions used in the two corollaries.  Write
$\mathrm{PG}(N,q)$ for the $N$-dimensional projective space over
$\mathbb F_q$, whose points are the one-dimensional subspaces of
$\mathbb F_q^{N+1}$.  For $2\leqslant N\leqslant q-2$, a
\emph{normal rational curve} in $\mathrm{PG}(N,q)$ is a $(q+1)$-arc
projectively equivalent to the points represented by
$(1,t,\ldots,t^N)$, $t\in\mathbb F_q$, together with $(0,\ldots,0,1)$
\cite[Sections~2--3]{BallLavrauw2019}.  It is \emph{complete} if no
further point can be adjoined while preserving the arc property.

For $2\leqslant k\leqslant q-2$, define the
\emph{projective Reed--Solomon code} by
\[
 \operatorname{PRS}_q(k):=
 \bigl\{(f(a))_{a\in\mathbb F_q\cup\{\infty\}}:
 f\in\mathbb F_q[X],\ \deg f<k\bigr\}
 \subseteq \mathbb F_q^{q+1},
\]
where $f(\infty)$ is the coefficient of $X^{k-1}$.  This is a linear code
of length $q+1$ and dimension $k$.  For a code
$C\subseteq\mathbb F_q^n$, write
\[
 \rho(C):=\max_{u\in\mathbb F_q^n}\min_{c\in C}d_H(u,c)
\]
for its covering radius, where $d_H(u,c)$ is the number of coordinates in
which $u$ and $c$ differ.  D{\"u}r \cite{Dur1994} proved that
$\rho(\operatorname{PRS}_q(k))=q-k$ if and only if the degree-$(q-k)$
normal rational curve in $\mathrm{PG}(q-k,q)$ is complete; see also
\cite[Lemma~II.4]{ZhangWanKaipa2020}.

For odd $q$, following
\cite[Section~1]{BartoliDavydovMarcuginiPambianco2018}, call a proper
subset $S$ of a nondegenerate conic
$\mathcal C\subseteq\mathrm{PG}(2,q)$ \emph{almost complete} if it can be
extended to a larger arc only by points of $\mathcal C\setminus S$; let
$t(q)$ be the minimum size of such a subset.

The proof of Corollary~\ref{cor:nrc-complete} joins a low-dimensional
completeness range to the high-dimensional range supplied by
Theorem~\ref{thm:mds-range}.  For odd nonprime $q\ne27$, Storme's theorem
controls the lower projective dimensions in terms of $t(q)$, while
Theorem~\ref{thm:mds-range} controls the upper dimensions;
Lemma~\ref{lem:tq-splice} shows that the two ranges meet without a gap.
Prime fields are already covered by the prime-field MDS theorem, and the case
$q=27$ is handled separately by the known bounds cited below.

By Storme's theorem \cite[Theorem~3.1]{Storme1992}, if an almost complete
subset has size $s$, then every normal rational curve in
$\mathrm{PG}(N,q)$ is complete whenever $2<N<q-s+2$.  We also use the upper
bound and the computer-search data of Bartoli, Davydov, Marcugini, and
Pambianco
\cite[Theorem~4.12 and Tables~2--3]{BartoliDavydovMarcuginiPambianco2018}.

\begin{lemma}\label{lem:tq-splice}
Let $q=p^h$ be an odd nonprime prime power with $q\ne27$.  Then
\[
 t(q)\leqslant B(p,q)+1.
\]
\end{lemma}

\begin{proof}
The comparison with $B(p,q)$ starts from
\begin{equation}\label{nrc:B-third}
 \frac{(p-2)q+6p-10}{2p-3}-\frac{q+8}{3}
 =\frac{(p-3)(q+2)}{3(2p-3)}\geqslant0.
\end{equation}
First suppose $q\geqslant243$.  The function $x/32-\ln x$ is increasing
for $x\geqslant243$ and positive at $x=243$, so $\ln q<q/32$.
By \cite[Theorem~4.12]{BartoliDavydovMarcuginiPambianco2018},
\[
 t(q)<1.835\sqrt{q\ln q}
 <\frac{1.835}{\sqrt{32}}q
 <\frac q3
 <\frac{q+8}{3}
 \leqslant\frac{(p-2)q+6p-10}{2p-3}.
\]
Since $t(q)$ is an integer, $t(q)\leqslant B(p,q)$.

It remains to consider the odd nonprime prime powers below $243$, with
$q\ne27$:
\[
 q\in\{9,25,49,81,121,125,169\}.
\]
Tables~2--3 of \cite{BartoliDavydovMarcuginiPambianco2018} provide
almost complete subsets of the following sizes $s(q)$.
\[
\begin{array}{c|rrrrrrr}
q        &9&25&49&81&121&125&169\\ \hline
s(q)     &6&12&18&25&33&35&41\\
B(p,q)+1 &6&14&26&30&61&57&84
\end{array}
\]
Therefore $t(q)\leqslant s(q)\leqslant B(p,q)+1$ in every remaining
case.
\end{proof}

\begin{proof}[Proof of Corollary~\ref{cor:nrc-complete}]
If $q$ is prime, the result follows from the prime-field case of the MDS
conjecture \cite[Theorem~1.7]{Ball2012}.  For $q=27$, every arc in
$\mathrm{PG}(N,27)$ has at most $28$ points for
$2\leqslant N\leqslant25$; see \cite{ChaoKaneta2001} and the summary
\cite[Table~3.7]{HirschfeldStorme2001}.  A normal rational curve already
has $28$ points, so it is complete in each of these dimensions.

Now assume that $q=p^h$ is nonprime and $q\ne27$.
Equation~\eqref{nrc:B-third} and $q\geqslant9$ give $B(p,q)\geqslant5$.
Theorem~\ref{thm:mds-range}, applied in vector dimension $3$, therefore
proves completeness for $N=2$.

Let $3\leqslant N\leqslant q-2$.  If $N<q+2-t(q)$, Storme's theorem
\cite[Theorem~3.1]{Storme1992} gives completeness.  Otherwise,
\[
 N\geqslant q+2-t(q)
 \geqslant q+1-B(p,q)
\]
by Lemma~\ref{lem:tq-splice}.  The corresponding vector dimension
$k=N+1$ therefore satisfies $k\geqslant q+2-B(p,q)$.  Theorem~\ref{thm:mds-range}
then excludes a $(q+2)$-arc in this dimension, so the normal rational curve
is complete.
\end{proof}

\begin{proof}[Proof of Corollary~\ref{cor:prs-covering-radius}]
By D{\"u}r's result \cite{Dur1994} (see also
\cite[Lemma~II.4]{ZhangWanKaipa2020}), for $2\leqslant k\leqslant q-2$
the covering radius of $\operatorname{PRS}_q(k)$ is $q-k$ exactly when the
degree-$(q-k)$ normal rational curve in $\mathrm{PG}(q-k,q)$ is complete.
Since $2\leqslant q-k\leqslant q-2$,
Corollary~\ref{cor:nrc-complete} applies.
\end{proof}

\noindent\textbf{AI disclosure.}
OpenAI's ChatGPT (GPT-5.6 Sol and GPT-6 Astra) was used to assist with
exploratory proof development and computations, literature searches, and
English-language polishing.  The author takes full responsibility for the
correctness and content of this paper.

\providecommand{\bysame}{\leavevmode\hbox to3em{\hrulefill}\thinspace}
\providecommand{\MR}{\relax\ifhmode\unskip\space\fi MR }
% \MRhref is called by the amsart/book/proc definition of \MR.
\providecommand{\MRhref}[2]{%
  \href{http://www.ams.org/mathscinet-getitem?mr=#1}{#2}
}
\providecommand{\href}[2]{#2}

\end{document}